\documentclass[11pt]{article}

\usepackage{graphicx} 
\usepackage[utf8]{inputenc}
\usepackage{amsmath}
\usepackage{amsfonts}
\usepackage{caption}
\usepackage{graphicx}
\usepackage{amssymb}
\usepackage{dirtytalk}
\usepackage[dvipsnames,table,xcdraw]{xcolor}
\usepackage{comment}
\usepackage{enumerate}   
\usepackage{bbm}
\usepackage{amsthm}
\usepackage{thmtools}
\usepackage[numbers,sort]{natbib}
\usepackage{comment}
\usepackage{algorithm}
\usepackage{algpseudocode}
\usepackage[margin=1in]{geometry}
\usepackage{pgfplots}
\usepackage{crossreftools}
\usepackage{multirow}
\usepackage{bm}
\usepackage{empheq}
\usepackage{microtype}
\usepackage{xurl}
\usepackage{hyperref}
\usepackage{cleveref} 
\usepackage{tabularx}
\usepackage{booktabs}

\pgfplotsset{compat=1.18}
\hypersetup{
    colorlinks=true,
    linkcolor=Red,
    citecolor=Cyan,
    urlcolor=Magenta
}

\makeatletter
\def\thmt@refnamewithcomma #1#2#3,#4,#5\@nil{%
  \@xa\def\csname\thmt@envname #1utorefname\endcsname{#3}%
  \ifcsname #2refname\endcsname
    \csname #2refname\expandafter\endcsname\expandafter{\thmt@envname}{#3}{#4}%
  \fi
}
\makeatother

\newtheorem{theorem}{Theorem}[section]

\newtheorem{claim}[theorem]{Claim}
\newtheorem{lemma}[theorem]{Lemma}
\newtheorem{observation}[theorem]{Observation}

\newtheorem{corollary}[theorem]{Corollary}
\newtheorem{proposition}[theorem]{Proposition}

\newcounter{algorithmicinstance}
\AddToHook{env/algorithmic/begin}{\stepcounter{algorithmicinstance}}
\expandafter\def\csname theHALG@line\endcsname{\arabic{algorithmicinstance}.\arabic{ALG@line}}
\makeatletter
\crefalias{ALG@line}{line}
\AddToHook{env/algorithmic/begin}{%
    \let\algorithmiclabel\label
    \def\label{%
        \def\@currentcounter{ALG@line}%
        \protected@edef\cref@currentlabel{[line][\arabic{ALG@line}][]\theALG@line}%
        \algorithmiclabel
    }%
}
\makeatother

\theoremstyle{definition}
\newtheorem{definition}{Definition}
\newtheorem{condition}{Condition}

\crefname{claim}{claim}{claims}
\crefname{inftheorem}{informal theorem}{informal theorems}
\crefname{figure}{figure}{figures}
\crefname{@theorem}{theorem}{theorems}
\crefname{Definition}{definition}{definitions}
\crefname{line}{line}{lines}

\Crefname{claim}{Claim}{Claims}
\Crefname{inftheorem}{Informal Theorem}{Informal Theorems}
\Crefname{@theorem}{Theorem}{Theorems}
\Crefname{Definition}{Definition}{Definitions}
\Crefname{figure}{Figure}{Figures}
\Crefname{line}{Line}{Lines}

\newcommand{\R}{\mathbb{R}}

\newcommand{\argmax}{\operatorname*{arg\,max}}
\newcommand{\E}{\mathbb{E}}
\newcommand{\D}{\mathcal{D}}
\newcommand{\eps}{\epsilon}

\newcommand{\Alg}{\mathsf{Alg}}

\newcommand{\Opt}{\mathsf{Opt}}
\newcommand{\SSSI}{\mathsf{SSPI}}
\newcommand{\kSSI}{k\text{-}\mathsf{SPI}}
\newcommand{\SSCPI}{\mathsf{SSCombPI}}
\newcommand{\CPI}{\mathsf{CombPI}}
\newcommand{\kSCPI}[1]{#1\text{-}\mathsf{SCombPI}}

\newcommand{\FD}{\mathsf{FD}}
\newcommand{\FDAlg}{\mathsf{FDAlg}}
\newcommand{\FracKnap}{\mathsf{FracKnap}}

\newcommand{\OSWM}{\mathsf{OSWM}}
\newcommand{\OFDA}{\mathsf{OFDA}}
\newcommand{\PAlg}{\mathsf{PAlloc}}

\newcommand{\va}{\bm{a}}
\newcommand{\vc}{\bm{c}}
\newcommand{\vsc}{\bm{s}^c}
\newcommand{\vp}{\bm{p}}
\newcommand{\vq}{\bm{q}}
\newcommand{\vr}{\bm{r}}
\newcommand{\vs}{\bm{s}}
\newcommand{\vu}{\bm{u}}
\newcommand{\vv}{\bm{v}}
\newcommand{\vx}{\bm{x}}
\newcommand{\vy}{\bm{y}}
\newcommand{\vz}{\bm{z}}
\newcommand{\valpha}{\bm{\alpha}}
\newcommand{\vbeta}{\bm{\beta}}
\newcommand{\ovl}{\overline}
\newcommand{\inner}[2]{\langle #1, #2 \rangle}

\allowdisplaybreaks

\algrenewcommand\algorithmiccomment[1]{\hfill {\color{gray} \(\triangleright\) \itshape #1}}
\algnewcommand{\LineComment}[1]{\item[] {\color{gray}  \(\triangleright\) \itshape #1}}
\algnewcommand\Input{\item[\textbf{Input:}]}
\algnewcommand\Output{\item[\textbf{Output:}]}\algnewcommand{\IfThenElse}[3]{% \IfThenElse{<if>}{<then>}{<else>}
  \State \algorithmicif\ #1\ \algorithmicthen\ #2\ \algorithmicelse\ #3}

\newcommand{\remove}[1]{}

\AddToHook{cmd/appendix/before}{%
    \crefalias{section}{appendix}%
    \crefalias{subsection}{appendix}
}

\title{Single-Sample Prophet Inequalities: \\ A Combinatorial to Single-Item Reduction}

\author{
Shuchi Chawla \\ \small The University of Texas at Austin \\ {\tt \small shuchi@cs.utexas.edu} \and 
Trung Dang  \\ \small The University of Texas at Austin \\ {\tt \small dddtrung@cs.utexas.edu}
}

\date{}

\begin{document}

\maketitle

\begin{abstract}
    We study single-sample prophet inequalities for online combinatorial allocation. Our main contribution is a general reduction from combinatorial to single-item prophet inequalities for valuation classes admitting suitable supporting prices. The reduction uses a free-disposal value to separate buyer-side combinatorial constraints from item-side supply constraints, yielding a modular framework that applies in the stronger Game of Googol model. This framework yields a $\frac{1}{6\sqrt{3}}\approx\frac{1}{10.4}$-competitive single-sample prophet inequality and a $(\beta_{k-1}/4)$-competitive $k$-sample prophet inequality for XOS valuations, where $\beta_k$ is the competitive ratio of a $k$-sample single-item prophet inequality, improving upon the work of \cite{dutting2024online}. Both results extend directly to divisible resources with capped-XOS valuations. Along the way, we obtain new results for online free disposal and an optimal single-sample prophet inequality for fractional knapsack in the Game of Googol model.
\end{abstract}

\newpage
\tableofcontents

\newpage

% \shuchi{Remaining to-dos: related work;  pass through intro; add all refs; }

\section{Introduction}

Can samples be as informationally powerful as the entire distribution? Single-sample prophet inequalities push this question to its informational extreme. An online algorithm faces an input drawn from an unknown distribution and must make irrevocable decisions as the instance is revealed. To guide these decisions, it receives just one independent sample from the underlying distribution. The central question is whether this tiny amount of information can nevertheless support guarantees comparable to those achievable when the full distribution is known.

In a striking result, Rubinstein, Wang, and Weinberg~\cite{rubinstein20samples} showed that, for the classical single-item prophet inequality, the answer is yes. In this problem, the algorithm observes $n$ nonnegative values drawn independently from possibly different distributions and presented in an adversarial order. It must immediately accept or reject each value, accepting at most one, with the goal of maximizing the expected value selected. Even with full knowledge of the distributions, no online algorithm can guarantee more than half the expected hindsight optimum in the worst case. Rubinstein et al. showed that this optimal factor-$2$ guarantee can be matched using just one sample from each distribution. The algorithm is remarkably simple: set a threshold equal to the largest sampled value, and accept the first real value that exceeds it.

Subsequent work has demonstrated the power of a single sample in a variety of online settings, including multi-unit selection, matching, and selection under several classes of matroid constraints \cite{caramanis2022single,kaplan2022online,pashkovich2025single,nuti2026static}. Many of these results concern {\em single-dimensional} selection, where each arriving candidate has a scalar value and the algorithm decides whether to accept or reject it subject to feasibility constraints. In {\em multidimensional} allocation, the algorithm must additionally decide how to serve each candidate, for example, which subset of items to allocate to a buyer whose value depends on the entire bundle. Single-sample guarantees are known for several such settings, including unit-demand and budget-additive valuations~\cite{caramanis2022single}, and D\"{u}tting et al.~\cite{dutting2024online} recently extended this line of work to general XOS (fractionally subadditive) valuations. These results, however, rely on constructions and analyses tailored to the particular valuation class, leaving open whether there is a common principle underlying single-sample prophet inequalities in multidimensional settings.

In this paper, we present a general reduction from multidimensional to single-dimensional single-sample prophet inequalities. Our framework applies whenever buyers' valuation functions admit suitable supporting prices, which allow the value of a multidimensional allocation to be represented through per-item contributions. This reduction lets us lift single-item guarantees to substantially richer combinatorial allocation problems.

To understand our reduction, consider combinatorial allocation with unit-demand buyers. There are $m$ items and $n$ buyers arriving online. Each buyer has a value for every item, with her $m$-dimensional value vector drawn from an arbitrary joint distribution. A tempting approach is to treat the items independently, running $m$ single-item single-sample prophet inequalities, one for each item. But this ignores the coupling among a buyer's values. A single buyer may simultaneously have high values for many items and hence be competitive in many of these prophet inequalities. If she participates in all of them, she may block other buyers from receiving several items, even though, being unit-demand, she ultimately derives value from only one. To avoid this loss, we must first decide which competition---that is, which item-specific prophet inequality---each buyer will participate in.

Our reduction formalizes precisely this separation. In the first stage, we choose a buyer-feasible \emph{provisional allocation}: for a unit-demand buyer, this amounts to selecting a single candidate item or a single competition to participate in. At this stage we ignore supply constraints, so many buyers may provisionally select the same item. The right objective for making these choices is a \emph{free-disposal} objective: for each item, we count only the most valuable provisional claim, since all the others may eventually be discarded. The second stage then runs a single-item prophet inequality independently for each item to resolve these competing claims. More generally, supporting prices extend this same decomposition beyond unit-demand valuations: the free-disposal algorithm handles the multidimensional, buyer-side constraints, while the single-item prophet inequalities enforce the item-side supply constraints.

Our framework also gives a substantially sharper understanding of the single-sample prophet inequality for XOS valuations of D\"{u}tting et al.~\cite{dutting2024online}, who first obtained a constant-factor guarantee for this problem. Their analysis gives a competitive ratio of $576$. Viewed through the free-disposal framework, essentially the same algorithm admits a much simpler and more modular analysis, improving the factor to $24$. More importantly, this viewpoint exposes further opportunities for improvement, both in the underlying free-disposal algorithm and in the analysis of the reduction. Exploiting these opportunities improves the factor further to $6\sqrt{3}\approx 10.4$. Beyond the numerical improvement, the framework cleanly separates the combinatorial and single-item ingredients of the problem, making clear where further advances can enter and enabling extensions to broader multidimensional settings.

\subsection{Technical Overview and Results}

\paragraph{Linearization of values.}
At a high level, our reduction from combinatorial to single-dimensional settings maps each buyer's combinatorial valuation function $\vv_i$ to a linear function $\vq_i$. The $j$-th coordinate $q_{i,j}$ specifies the value with which buyer $i$ enters the competition for item $j$. Each item is ultimately allocated to at most one of the buyers competing for it, and the resulting linear objective provides a lower bound on the buyers' true welfare once the item-wise allocations are combined. Thus, the linearization aims to serve two purposes: it determines which item-wise competitions each buyer participates in, and it decouples the objectives for these competitions while preserving a lower bound on the welfare of the resulting allocation.

We instantiate the linearization using supporting prices. For example, for an XOS valuation $v$, every allocation $\vy$ admits a supporting-price vector $q^{(\vy)}$ such that
\[
    v(\vy)=\inner{q^{(\vy)}}{\vy},
    \qquad\text{and}\qquad
    v(\vx)\ge \inner{q^{(\vy)}}{\vx}
\]
for every $\vx\le \vy$. The allocation $\vy_i$ at which $\vq_i$ supports buyer $i$'s valuation serves as the anchor for this linearization: the final allocation to buyer $i$ is obtained by trimming $\vy_i$. We refer to $\vy_i$ as the buyer's \emph{provisional allocation}.

\paragraph{The free-disposal value.}
%The key question is then how to choose these linearizations so that a large fraction of the hindsight-optimal welfare can still be recovered after enforcing the supply constraints. 
We cast the choice of linearizations as an online optimization problem. As each buyer arrives, we choose a provisional allocation $\vy_i$ and the corresponding supporting-price vector $\vq_i$ at $\vy_i$. The provisional allocations satisfy the buyer-side constraints, but may conflict with one another on the supply side. The goal is to ensure that, after resolving these conflicts, the supporting-price value that can be retained is comparable to the hindsight-optimal social welfare. We call this retained value the \emph{free-disposal value}.\footnote{We use the term \emph{free disposal} in the same spirit as Feldman et al.~\cite{FeldmanEtAl09} and Charikar et al.~\cite{CharikarEtAl14}. In \cite{FeldmanEtAl09}, in particular, items arrive online, so item-side supply constraints are automatically respected, while buyer-side demand constraints may be violated and excess allocation is discarded. Our setting reverses the roles of the two sides: buyers arrive online, so we enforce buyer-side constraints when choosing provisional allocations, while allowing item-side supply constraints to be violated; free disposal then trims the provisional allocations back to supply feasibility.}

%The term free disposal is used in the literature in several related senses. For example, Feldman et al.~and Devanur et al.~use it in settings where buyers may initially be allocated beyond their demands and excess allocation is subsequently discarded. Our relaxation is in the opposite direction: the provisional allocations respect buyer-side demand constraints but may violate item-side supply constraints. For unit-demand valuations, our use is closer to the free-disposal matching model of Charikar et al. \shuchi{Add citations.}}

It remains to understand whether the free-disposal problem itself admits good online algorithms. We show that, for the valuation classes we consider, its objective can be represented as an online monotone submodular maximization problem subject to a simple partition matroid constraint. In particular, the familiar buyer-wise greedy algorithm for XOS valuations \cite{dobzinski2005approximation} is exactly the greedy algorithm for this formulation, giving a constant-factor guarantee for the free-disposal value. This provides the first component of our reduction; the second component consists of the item-wise single-dimensional prophet inequalities.

\paragraph{The role of samples.} In our reduction, samples serve two distinct purposes: coordinating provisional allocations and supplying the inputs used by the item-wise prophet inequalities. Although the underlying worst-case free-disposal algorithm needs no distributional information, its choices generally depend on earlier buyers' realized valuations. Feeding these choices directly into the item-wise prophet inequalities would introduce dependencies across buyers and invalidate their guarantees. We avoid this difficulty using an additional {\em coordination sample}. We run the free-disposal algorithm on this sample and record its state immediately before processing each buyer. When real buyer $i$ arrives, we return the provisional allocation and supporting prices that the algorithm would produce if $i$'s real valuation replaced her sampled valuation at that recorded state. Crucially, this counterfactual execution does not use any other buyer's real valuation. Conditional on the coordination sample and shared internal randomness, each buyer's output depends only on her own valuation. We call this property {\em buyer-locality}; it preserves independence across buyers, while allowing dependencies across items belonging to the same buyer.

Combining these components, an $\alpha$-competitive free-disposal algorithm and a $\beta_k$-competitive $k$-sample single-item prophet inequality yield an $\alpha\beta_k/2$-competitive $(k+1)$-sample combinatorial prophet inequality for any $k\ge 1$ (\Cref{cor:k-plus-1-sscpi-to-k-sspi}).

\paragraph{Random-order free disposal.} Buyer-locality also separates the real arrival order of buyers from the order over sampled values used by the free-disposal simulation. This allows us to process the coordination sample in a uniformly random order, even when the real buyers arrive adversarially. Doing so yields a modest improvement in the competitive ratio.

\paragraph{Obtaining a truly single-sample PI.}
Finally, we adapt the sample-splitting idea of D\"utting et al.~\cite{dutting2024online} to our reduction: a single sample is randomly divided between the coordination and competition roles, allowing us to construct a single-sample combinatorial prophet inequality from a single-sample single-item one. Our contribution is a sharper analysis of this splitting step. By tracking the performance of the free-disposal algorithm as a function of the split probability, we can optimize this probability rather than fixing a balanced split, leading to the improved guarantees reported below.

\paragraph{Game of Googol setting.}
All of our prophet guarantees hold in the stronger Game of Googol (GoG) model \cite{correa2020googol}. In this setting, there is no underlying distribution over instances. Instead, for each buyer, an adversary chooses $k+1$ valuation functions, which are then randomly assigned to $k$ sample positions and one real position. The samples are revealed to the algorithm before the process begins, while the real valuations arrive online in adversarial order. Since the GoG model is stronger than the standard sample-based setting, every guarantee we establish in the GoG model immediately applies to the latter as well.

\paragraph{Divisible items.} The key technical ingredient underlying our results is the existence of one-sided supporting prices. For indivisible-item settings, this condition characterizes XOS valuations. For divisible items, it captures the strictly larger class of \emph{capped-XOS} valuations. A capped-XOS valuation is a maximum of capped-additive clauses of the form
\[
    h_{a,c}(x)=\sum_j a_j\min\{x_j,c_j\},
    \qquad a_j\geq 0,\quad c_j\in[0,1].
\]
In the divisible-item setting, for each item $j$, buyer $i$'s provisional allocation contributes an amount $y_{i,j}$ at density $q_{i,j}$. Our continuous free-disposal algorithm for this class of valuations has the same smooth guarantee as its indivisible counterpart. Resolving supply competition, on the other hand, is equivalent to solving a fractional knapsack with capacity one. We show that this problem admits the same optimal $1/2$ single-sample prophet inequality in the GoG setting as the classical single-choice counterpart. Consequently, we obtain the same competitive ratio in both settings.

\paragraph{Our results.}
\Cref{tab:overview} summarizes the resulting single-sample guarantees. Competitive ratios are stated as $<1$, so larger is better. All our guarantees hold in the Game of Googol model. The XOS algorithm uses polynomially many demand and supporting-price queries. For divisible capped-XOS valuations, the free-disposal optimization has an efficient implementation under the additional concavity and oracle assumptions specified in \Cref{sec:free-disp}. The subadditive guarantee follows from the standard $H_m$-approximation by XOS valuations~\cite{dobzinski2007two,bhawalkar2011welfare}.

\begin{table}[hbt]
\centering
\small
\renewcommand{\arraystretch}{1.3}
\setlength{\tabcolsep}{4pt}
\begin{tabularx}{\linewidth}{@{\extracolsep{\fill}}>{\hsize=0.9\hsize\linewidth=\hsize\raggedright\arraybackslash}X >{\hsize=1.05\hsize\linewidth=\hsize\raggedright\arraybackslash}X >{\hsize=1.05\hsize\linewidth=\hsize\raggedright\arraybackslash}X@{}}
\toprule
\textbf{Setting} & \textbf{Previous single-sample guarantee} & \textbf{This work (GoG)}\\
\midrule
Indivisible items:\newline XOS
& $1/576$\newline\cite{dutting2024online}
& $1/(6\sqrt{3})\approx0.0962$\newline(\Cref{thm:sscpi-xos})\\
\addlinespace
Indivisible items:\newline monotone subadditive
& $1/(576H_m)$, implied by\newline\cite{dutting2024online,dobzinski2007two,bhawalkar2011welfare}
& $1/(6\sqrt{3}H_m)$\newline(\Cref{thm:sscpi-xos,thm:subadd-sspi-reduction})\\
\addlinespace
Divisible items:\newline capped-XOS
& N/A
& $1/(6\sqrt{3})\approx0.0962$\newline(\Cref{thm:sscpi-capped-xos})\\
\addlinespace
One divisible item:\newline fractional knapsack
& $1/2$ for independent draws\newline\cite{nuti2026sample}
& $1/2$, optimal\newline(\Cref{thm:sspi-frac-knapsack})\\
\addlinespace
One divisible item:\newline monotone concave valuations
& N/A
& $(1-1/e)/2-\varepsilon$\newline(\Cref{thm:sspi-concave})\\
\bottomrule
\end{tabularx}
\caption{Single-sample guarantees for adversarial real arrivals. Here $H_m=\sum_{j=1}^m1/j$ denotes the $m$-th harmonic number and $\varepsilon>0$ can be arbitrarily small.}
\label{tab:overview}
\end{table}

\paragraph{Other notable results.} Our results require several other new ingredients that are of independent interest. First, we prove an optimal $1/2$ single-sample prophet inequality for fractional knapsack in the Game of Googol model. The algorithm is the random-quantile threshold rule of Nuti et al.~\cite{nuti2026sample}. Our contribution is an analysis that also handles the sample-real dependence in the Game of Googol model. Second, beyond the greedy algorithms described above, we give a $(1-1/e)$-competitive free-disposal algorithm for concave DR-submodular valuations over divisible items, even under adversarial arrivals. Finally, for a single divisible item with monotone concave valuations, we give a separate algorithm based on discretizing marginal values with a competitive ratio approaching $0.316$.

\subsection{Related Work}
\label{sec:related-work}

\paragraph{Sample-based prophet inequalities.}
Prophet inequalities with limited distributional information were initiated by Azar, Kleinberg, and Weinberg \cite{azar2014prophet}, who showed that a single sample suffices for constant-factor guarantees in a number of combinatorial selection problems and gave a general reduction from certain order-oblivious secretary algorithms to single-sample prophet inequalities. Rubinstein, Wang, and Weinberg~\cite{rubinstein20samples} subsequently showed that, for the classical single-choice problem, a particularly simple single-sample threshold rule achieves the optimal $1/2$ guarantee obtainable even with complete knowledge of the distributions. A constant number of samples also suffices to get arbitrarily close to the optimal full-information guarantee for the i.i.d., random-order, and free-order variants~\cite{rubinstein20samples,guo2021generalizing,correa2023constant}. Meanwhile, Correa, Cristi, Epstein, and Soto~\cite{correa2020googol} studied the closely related two-sided Game of Googol, making explicit the sample--real symmetry that also underlies our work.

A subsequent line of work has developed direct single-sample algorithms for richer feasibility constraints. Caramanis et al.~\cite{caramanis2022single} introduced a greedy-ordered selection framework that yields improved guarantees for several matroid and matching settings and gives a single-sample algorithm for budget-additive combinatorial auctions. Closely related work of Kaplan, Naori, and Raz~\cite{kaplan2022online} studies online weighted matching with a sample and develops connections between random-order algorithms, adversarial-order-with-a-sample algorithms, and single-sample prophet inequalities. More recently, single-sample guarantees have been studied for uniform matroids, multi-unit selection, and various knapsack constraints \cite{pashkovich2025single,nuti2026static,nuti2026sample}. Our work differs from these results in that the objects arriving online are themselves multidimensional valuation functions: rather than selecting or rejecting a scalar-valued candidate, the algorithm must decide how to allocate several resources to each buyer.

\paragraph{Combinatorial prophet inequalities.}
There is an extensive literature on prophet inequalities for combinatorial allocation when the underlying distributions are known. Feldman, Gravin, and Lucier~\cite{feldman2014combinatorial} showed that static item prices obtain a $1/2$ fraction of the optimal expected welfare for XOS valuations. D\"utting, Feldman, Kesselheim, and Lucier~\cite{dutting2020prophet} developed a general balanced-pricing framework that unifies and extends a number of such prophet-inequality and posted-price results. These works exploit additive prices or supporting prices to summarize the contribution of individual resources to a combinatorial allocation. Supporting prices play a related but different role in our reduction: they allow a provisional multidimensional allocation to be decomposed into single-resource competitions while certifying the value of any allocation obtained by subsequently discarding resources. Other related work has studied random-order arrivals~\cite{ehsani2018prophet} and broader valuation classes like subadditive valuations~\cite{dutting2020log,correa2023constant}.

When the distributions are unknown, obtaining comparable guarantees from very few samples is substantially more difficult. Prior single-sample results for multidimensional allocation were known for special valuation classes such as unit-demand and budget-additive valuations~\cite{kaplan2022online,caramanis2022single}. Most directly related to the present work, D\"utting et al.~\cite{dutting2024online} gave the first constant-factor single-sample prophet inequality for general XOS valuations, obtaining a $1/576$ guarantee in the pointwise, or Game of Googol, model. Their algorithm uses buyer-wise greedy with reallocations together with a secretary-style ``compete against yourself'' argument. The same work also shows that with polynomially many samples one can obtain nearly the $1/2$ full-information guarantee through learned posted prices.

% \trung{this paragraph repeats the last one in the intro}

% Our work gives a different interpretation of the single-sample XOS problem. We isolate an online \emph{free-disposal} problem that is responsible for handling the buyers' combinatorial constraints, and use independent single-item prophet inequalities to enforce the supply constraint for each resource. Under this viewpoint, the buyer-wise greedy procedure underlying~\cite{dutting2024online} is precisely an online algorithm for the free-disposal problem. This separation yields a substantially simpler analysis of essentially their algorithm, improving the approximation factor from $576$ to $24$, and it makes the sources of further loss explicit. We then exploit this modularity in two ways: algorithmically, by running the free-disposal algorithm in an internal random order, and analytically, through the smooth competitive ratio that captures its performance when only a random subset of buyers is available. Together with a more careful division of buyers between coordination and competition roles, these ideas improve the factor further to $6\sqrt{3}\approx 10.4$.

% \shuchi{Mention subadditive.}

\paragraph{Free disposal and online welfare maximization.}
The buyer-wise greedy algorithm we use for XOS valuations originates in work of Dobzinski, Nisan, and Schapira~\cite{dobzinski2005approximation} on welfare maximization for complement-free valuations and also forms a central ingredient of~\cite{dutting2024online}. Our analysis observes that its free-disposal value can be represented as a monotone submodular function subject to a partition-matroid constraint. This allows us to invoke the random-order greedy analysis of Buchbinder, Feldman, Filmus, and Garg~\cite{buchbinder2020online}, including the stronger prefix-type guarantee that leads to our notion of smooth competitiveness. In the opposite direction, our hardness result for adversarial-order XOS free disposal builds on the $1/2$ hardness for online submodular welfare maximization of Kapralov, Post, and Vondr\'ak~\cite{kapralov2013online}. Thus, random order enters our work only inside the auxiliary free-disposal problem; the actual buyers in the prophet-inequality instance may still arrive adversarially.

\paragraph{Prophet inequalities with divisible resources.}
Prophet inequalities for divisible resources have received comparatively less attention. Oliveira, Singh, and Singla~\cite{oliveira2026philosopher} recently studied online welfare maximization with multidimensional concave valuations over divisible resources and obtained a tight $1/2$ prophet inequality under distributional knowledge using Aumann--Shapley supporting prices. Their use of supporting prices is closely related to the continuous supporting-price representation in our work. On the sample-based side, Nuti, Reiffenh\"auser, and Sayutina~\cite{nuti2026sample} obtain an optimal $1/2$ single-sample guarantee for fractional knapsack in the standard independent-draw model. We prove the same guarantee in the stronger Game of Googol model and use fractional knapsack as the single-resource primitive in our reduction. This yields single-sample combinatorial prophet inequalities for capped-XOS valuations over divisible resources, the continuous analogue of the XOS setting considered for indivisible items.

\subsection*{Organization of the paper} \Cref{sec:prelims} introduces the model, the classes of valuations we study, and the free-disposal value. We present the main reduction from $(k+1)$-sample combinatorial to $k$-sample single-item prophet inequalities in \Cref{sec:general-reduction}, along with the two-sample-to-one-sample improvement to obtain a single-sample combinatorial prophet inequality. \Cref{sec:free-disp} develops the required free-disposal algorithms for indivisible and divisible resources. \Cref{sec:sssi} studies the corresponding single-item single-sample prophet inequalities, including fractional knapsack and concave divisible valuations. Finally, \Cref{sec:final-results} combines these ingredients to obtain our main results; additional proofs and extensions appear in the appendices.

\subsection*{Acknowledgements and AI Disclosure} We are grateful to Sahil Singla for comments and suggestions on an earlier draft of this work. The authors used ChatGPT 5.6 Pro to assist with writing and completing the proofs. The authors have validated all details and take full responsibility for the content of this paper.

\section{Preliminaries} \label{sec:prelims}

% - online combinatorial allocation or online combinatorial auctions; combo PI

% - k sample PI - kSPI

% - Game of Googol - k-GoG-PI

% - budget constraint -> supply constraint

\subsection{Online Combinatorial Allocation, Prophet Inequalities, and the Game of Googol}

In the {\em Combinatorial Allocation} problem, there are $n$ buyers and $m$ items. We use $A\subseteq [0,1]^m$ to denote the space of feasible allocations over the $m$ items to an individual buyer. Each buyer $i\in [n]$ has a valuation function $v_i : A^m \to \R_{\ge 0}$ that maps a feasible allocation vector to a value. The function $v_i$ is drawn independently from an underlying distribution $\D_i$. Let $\D=\times_i\D_i$. Our goal is to find a partition of the $m$ items into allocations for each buyer that maximizes the social welfare, that is, the sum of the values received by each buyer. The optimal social welfare (a.k.a.\ hindsight optimum) is given by\footnote{We use $\bm{1}$ to denote the all-ones vector $1^m$.}
\[
\Opt(\D) := \E_{\vv\sim\D}[\Opt(\vv)] := \E\left[\max_{\substack{\vx_1, \dots, \vx_n \in A^m \\ \sum_{i=1}^n \vx_i \le \bm{1}}} \sum_{i=1}^n v_i(\vx_i)\right].
\]
We call the constraint $\vx_i\in A$ the demand constraint as it is a buyer constraint, and $\sum_{i=1}^n \vx_i \le \bm{1}$ the supply constraint as it is an item constraint.
\paragraph{Online allocation and competitive ratio.} In the online setting, buyers arrive sequentially. Upon the arrival of buyer $i$, the valuation $v_i$ is revealed to the algorithm, and the algorithm must irrevocably allocate a vector $\vx_i \in A^m$ to the buyer. For feasibility, the algorithm must satisfy $\sum_i\vx_i\le\bm{1}$. We write $\Alg(\vv) := \sum_{i=1}^n v_i(\vx_i)$ for the reward of the algorithm on instance $\vv$. We say that an algorithm is \emph{$\alpha$-competitive} for some $\alpha \in [0,1]$ if, for every product distribution $\D$, we have
\[\Alg(\D) := \E_{\vv\sim\D}[\Alg(\vv)] \ge \alpha \cdot \Opt(\D).\]
The optimal algorithm is also known as the {\em prophet}, and in the setting where the distribution $\D$ is known to the online algorithm, an $\alpha$-competitive algorithm is also called a {\em Combinatorial Prophet Inequality} (henceforth, $\CPI$).

\paragraph{Sample-based prophet inequalities.} In the sample-based setting, the distribution $\D$ is unknown to the algorithm, but the algorithm is supplied with $k$ samples drawn i.i.d.\ from $\D$ at the beginning of the online process. The algorithm may use these samples to coordinate online allocations across buyers. Note that each sample is a full instance composed of valuation functions for each buyer. We use $s_i \sim \D_i$ to denote the $i$-th sample, and $\vr\sim\D$ to denote the real instance over which the algorithm is run and evaluated.

\paragraph{Game of Googol.} Online allocation algorithms in the sample-based setting often exploit the symmetry between the samples $s_i$ and the real instance $r$, as both are drawn independently from the same underlying distribution $\D$. The Game of Googol (GoG) setting makes this symmetry explicit without requiring the existence of an underlying distribution. In this setting, an adversary picks $k+1$ instances $\vv^{(1)}, \cdots, \vv^{(k+1)}$ where $\vv^{(j)}=(v_i^{(j)})_{i\in [n]}$ for $j\in [k+1]$. For every buyer $i\in [n]$, the $k+1$ valuation vectors $v_i^{(1)}, \cdots, v_i^{(k+1)}$ are then randomly permuted and assigned to each of $k$ sample valuations $s_i^{(1)}, \cdots, s_i^{(k)}$ and one real valuation $r_i$. The algorithm is presented with the $k$ samples $\vs^{(1)}, \cdots, \vs^{(k)}$ where $\vs^{(j)}=(s_i^{(j)})_{i\in [n]}$. It is then run over the instance $\vr=(r_i)_{i\in [n]}$ in an online fashion. The algorithm's performance is compared against the hindsight-optimal solution over $\vr$.

We say that an algorithm achieves an $\alpha$-competitive ratio in the Game of Googol setting if
\[\E_{\vs, \vr}[\Alg(\vs, \vr)] = \E\left[\sum_{i=1}^n r_i(\vx_i)\right] \ge \alpha \cdot \E[\Opt(\vr)],\]
where the expectations are taken over the random assignment of the $k+1$ valuation vectors $v_i^{(1)}, \cdots, v_i^{(k+1)}$ to the sample and real valuations for each $i\in [n]$.

Note that the GoG setting is stronger than the sample-based setting in that an algorithm that is $\alpha$-competitive in the former setting is also $\alpha$-competitive in the latter. This work focuses primarily on the GoG setting. 

We use the term {\em $k$-Sample Combinatorial Prophet Inequality} (henceforth, $\kSCPI{k}$) to denote a competitive online algorithm in the GoG setting. The special case of $k=1$ is called a {\em Single-Sample Combinatorial Prophet Inequality} (henceforth, $\SSCPI$).

%In this problem, each buyer has two unknown, adversarially chosen valuation functions $v_{i, 1}, v_{i, 2}$. An independent fair coin is then flipped for each buyer to determine which valuation function serves as the sample valuation and which serves as the real valuation: either $(s_i, r_i) = (v_{i, 1}, v_{i, 2})$ or vice versa, with equal probability. This setting is strictly more general than the i.i.d.\ single-sample setting above, as we can let $v_{i, 1}$ and $v_{i, 2}$ be drawn independently from $\D_i$. As above, an algorithm $\Alg$ in this setting is $\alpha$-competitive if, for any choice of $v_{i, 1}, v_{i, 2}$ for all buyers,
%\[\E_{\vs, \vr}[\Alg(\vs, \vr)] = \E\left[\sum_{i=1}^n r_i(\vx_i)\right] \ge \alpha \cdot \E[\Opt(\vr)],\]
%where the $\vx_i$ are the allocation choices of the algorithm. Note that an $\alpha$-competitive algorithm for the Game of Googol problem is also $\alpha$-competitive for the single-sample online allocation problem.

%We also define the $k$-sample Game of Googol variant analogously, where each buyer has $k + 1$ adversarially chosen valuation functions $v_{i, 1}, \dots, v_{i, k + 1}$, and these valuations are permuted uniformly at random into the $k$ samples $s^1_i, \dots, s^k_i$ and the real valuation $r_i$. We use the acronyms $\SSCPI$ and $\kSCPI{k}$ to refer to the single-sample and $k$-sample variants of this problem, respectively.

\subsection{Assumptions on Valuation Functions}

In this work, we consider two main settings for the possible allocations to each buyer:
\begin{enumerate}
    \item \textbf{Indivisible items:} Each buyer can receive a discrete subset of the items, i.e., $A = \{0, 1\}$.
    \item \textbf{Divisible items:} Each buyer can receive a fractional allocation of the items, i.e., $A = [0, 1]$.
\end{enumerate}

Let $v$ be any valuation function considered. Throughout this paper, we assume that $v$ is \emph{normalized} and \emph{monotone}, i.e.,
$v(\bm{0}) = 0$ and $v(\vx) \le v(\vx')$ if $\vx \le \vx'$.\footnote{For vectors $\vx$ and $\vx'$, we use $\vx\le\vx'$ to denote the component-wise inequality $x_j\le x'_j$ for all $j\in [m]$.} The normalization condition is without loss of generality, as an algorithm that is $\alpha$-competitive for the normalized valuation $v'(\vx) = v(\vx) - v(\bm{0})$ is also $\alpha$-competitive for $v$. We also make the following crucial assumption on $v$, which underlies all of our techniques.

\begin{definition}[Supporting prices] \label{def:supporting-prices}
A valuation function $v : A^m \to \R_{\ge 0}$ (for $A = \{0, 1\}$ or $[0, 1]$) is said to have \emph{supporting prices} if, for any allocation $\vx \in A^m$, there exists a vector $\vq_v(\vx) \in \R_{\ge 0}^m$ such that
\begin{itemize}
    \item $\inner{\vq_v(\vx)}{\vx} = v(\vx)$.
    \item For all $\vx' \le \vx$, $\inner{\vq_v(\vx)}{\vx'} \le v(\vx')$.
\end{itemize}
When $v_i$ or $v$ is clear from context, we will simply use $\vq_i(\vx)$ or $\vq(\vx)$.
\end{definition}

Observe that this definition is different from a standard definition of supporting prices, which requires the inequality $\inner{\vq_v(\vx)}{\vx'} \le v(\vx')$ to hold for {\em all} $\vx'$ and not just those component-wise smaller than $\vx$. For indivisible items, the two definitions are equivalent, but for divisible items, the one-sided definition used here is weaker (more permissive) as discussed below.

Whenever computational efficiency is discussed, we assume access to a supporting-price oracle that, given $v$ and $\vx$, returns a valid supporting-price vector $\vq_v(\vx)$ in polynomial time.

\paragraph{Indivisible items.} In the indivisible-item setting, the supporting-price condition above is equivalent to the statement that $v$ is an \emph{XOS} valuation function.

\begin{definition}[XOS valuation function \cite{feige2009maximizing}]
    A valuation function $v : 2^{[m]} \to \R_{\ge 0}$ is \emph{XOS} if there exist vectors $\va_1, \dots, \va_K \in \R_{\ge 0}^m$ such that, for all $S \subseteq [m]$,
    \[v(S) = \max_{k \in [K]} \sum_{i \in S} a_{i, k}.\]
\end{definition}

We also consider two special subclasses of XOS valuations, namely \emph{submodular} and \emph{unit-demand} valuations.

\begin{definition}[Submodular valuation function]
    A valuation function $v : 2^{[m]} \to \R_{\ge 0}$ is \emph{submodular} if for any $S \subseteq T \subseteq [m]$ and $j \in [m]$,
    \[v(S \cup \{j\}) - v(S) \ge v(T \cup \{j\}) - v(T).\]
\end{definition}

\begin{definition}[Unit-demand valuation function]
    A valuation function $v : 2^{[m]} \to \R_{\ge 0}$ is \emph{unit-demand} if there exists a vector $\vv \in \R_{\ge 0}^m$ such that, for any $S \subseteq [m]$,
    \[v(S) = \max_{j \in S} v_j.\]
\end{definition}

\paragraph{Divisible items.} In the divisible-item setting, the supporting-price assumption is equivalent to the following notion of a \emph{capped-XOS valuation function}. We prove this equivalence in \Cref{app:prelims}.

\begin{definition}[Capped-XOS valuation function]
A \emph{capped-additive clause} is a function of the form
\[
h_{\va,\vc}(\vx)=\sum_{i=1}^m a_i\min\{x_i,c_i\},
\]
where $\va\in\R_{\ge 0}^m$ and $\vc\in[0,1]^m$. A valuation $v$ is \emph{capped-XOS} if there is a family $\{h_\ell\}_{\ell\in L}$ of capped-additive clauses such that for all $\vx \in [0, 1]^m$,
\[
v(\vx)=\max_{\ell\in L} h_\ell(\vx).
\]
\end{definition}

\begin{proposition} \label{prop:capped-XOS-sup-price}
A normalized, monotone valuation $v:[0,1]^m\to\R_{\ge 0}$ admits a supporting-price vector $\vq(\vx)$ at every $\vx\in[0,1]^m$ if and only if $v$ is capped-XOS.
\end{proposition}

We also consider a specific subclass of capped-XOS valuations, called \emph{DR-submodular} valuations. This subclass has been studied in previous work on prophet inequalities for divisible items \cite{oliveira2026philosopher}. In addition to having supporting prices, such valuations have the desirable property that their supporting-price map is continuous.

\begin{definition}[DR-submodular valuation function \cite{bian2017continuous,bian2020continuous}]
    A twice-differentiable valuation function $v : \R_{\ge 0}^m \to \R_{\ge 0}$ is \emph{DR-submodular} if it satisfies one of the following three equivalent conditions:
    \begin{enumerate}
        \item $v(\vy + \vz) - v(\vy) \ge v(\vx + \vz) - v(\vx)$ for $\vx, \vy, \vz \in \R_{\ge 0}^m$ and $\vy \le \vx$. 
        \item For any $\vx, \vy \in \R_{\ge 0}^m$ where $\vx \le \vy$, we have $\nabla v(\vx) \ge \nabla v(\vy)$.
        \item $\frac{\partial^2}{\partial x_i \partial x_j} v(\vx) \le 0$ for all $\vx \in \R_{\ge 0}^m$ and $i, j \in [m]$.
        \end{enumerate}
\end{definition}

\begin{proposition}[Lemma 4.2 of \cite{oliveira2026philosopher}] \label{prop:dr-submod}
For a DR-submodular function $v$, define the Aumann--Shapley \cite{shapley74values} vector $\vq(\vx)$ for all $\vx \in [0, 1]^m$ by
\[\vq(\vx) := \int_0^1 \nabla v(t \vx) \, dt.\]
Then $\vq$ is a valid continuous supporting-price map for $v$.
\end{proposition}

\subsection{Free-Disposal Value}

Under the assumption that valuation functions have supporting prices, we can also consider the \emph{free-disposal value} over the allocations $\vx_1, \dots, \vx_n$. Intuitively, for allocation vectors that do not satisfy the supply constraint $\sum_i \vx_i\le\bm{1}$, the free-disposal value is the best ``trimming'' of the allocations that satisfies the constraint. The trimming is evaluated with respect to the supporting prices of the original allocations. Formally, we define the objective as follows.

\begin{definition}[Free-disposal value]
For valuation functions $v_1, \dots, v_n$ with supporting-price maps $\vq_1, \dots, \vq_n$ and provisional allocation vectors $\vx_1, \dots, \vx_n \in A^m$ that do not necessarily satisfy the supply constraint, the \emph{free-disposal value} is defined by the following optimization program.
\begin{align*}
    \FD(\vx_1, \dots, \vx_n) &:= \max_{\substack{\vy_1, \dots, \vy_n \in A^m \\ \vy_i \le \vx_i \, \forall i \in [n]  \\ \sum_{i=1}^n \vy_i \le \bm{1}}} \sum_{i=1}^n \inner{\vq_i(\vx_i)}{\vy_i}.
\end{align*}
\end{definition}

\paragraph{Competitiveness in the free-disposal setting.} The input to an online algorithm in the free-disposal setting is the same as the input in a sample-based or GoG combinatorial prophet inequality. However, the algorithm's objective is to find (supply-infeasible) provisional allocations maximizing the free-disposal value. We say that an algorithm in this setting is $\alpha$-competitive if the free-disposal value of its provisional allocation is $\alpha$-competitive with respect to the hindsight optimum $\Opt$ defined previously.

\paragraph{Per-item decomposition.}
We observe that the free-disposal value is \emph{decomposable} into multiple per-item \emph{hindsight optima}.

\begin{observation}
    For any valuation functions $v_1, \dots, v_n$ with supporting prices $\vq_1, \dots, \vq_n$, and allocation vectors $\vx_1, \dots, \vx_n \in A^m$, in the divisible-item setting, we have
%    
    %we have
    %\[\FD(\vx_1, \dots, \vx_n) = \sum_{j=1}^m \Opt\left(\{(q_{i, j}(\vx_i), x_{i, j})\}_{i \in [n]}\right).\]
%In particular, if the items are divisible, then
\[\FD(\vx_1, \dots, \vx_n) = \sum_{j=1}^m \Opt\left(\{(q_{i, j}(\vx_i), x_{i, j})\}_{i \in [n]}\right) = \sum_{j=1}^m \FracKnap\left(\{(q_{i, j}(\vx_i), x_{i, j})\}_{i \in [n]}\right).\]
Here $\FracKnap\left(\{q_{i, j}(\vx_i), x_{i, j}\}_{i \in [n]}\right)$ is the value of the optimum fractional knapsack of size $1$ with $n$ items, where item $i$ has density $q_{i, j}(\vx_i)$ and weight $x_{i, j}$. If the items are indivisible, then we have
\[\FD(\vx_1, \dots, \vx_n) = \sum_{j=1}^m \Opt\left(\{(q_{i, j}(\vx_i), x_{i, j})\}_{i \in [n]}\right) = \sum_{j=1}^m \max_{i : x_{i, j} = 1} q_{i, j}(\vx_i).\]
\end{observation}

We also observe that the free-disposal value admits the following per-item \emph{layer-cake representation}.

\begin{proposition}
For valuation functions $v_1, \dots, v_n$ with supporting-price maps $\vq_1, \dots, \vq_n$, and allocation vectors $\vx_1, \dots, \vx_n \in A^m$, we have
\[\FD(\vx_1, \dots, \vx_n) = \sum_{j=1}^m \int_0^\infty \min\left\{1, \sum_{i=1}^n x_{i, j} \cdot \bm{1}[q_{i, j}(\vx_i) \ge w] \right\} \,dw.\]
\end{proposition}

\subsection{Single-Item Prophet Inequality}

We also consider the \emph{single-item online allocation problem} in this work. This is a special case of the general combinatorial allocation problem in which $m = 1$. In particular, we focus on a variant of the problem called the \emph{uniform-reward limited-demand variant}. In this variant, a buyer's valuation is a pair $(q_i, x_i)$, where $q_i \in \R_{\ge 0}$ is the buyer's unit reward for this item, while $x_i \in A$ is the buyer's demand for the item. The algorithm is given these pairs sequentially and, at each step, must irrevocably choose an allocation $x'_i$ for buyer $i$, where $x'_i \in A$ and $0 \le x'_i \le x_i$. As usual, the algorithm must satisfy supply feasibility, $\sum_{i=1}^n x'_i \le 1$, while its reward is $\sum_{i=1}^n q_i x'_i$.

We note that in the indivisible case, this variant is exactly the classical prophet inequality, with the caveat that a buyer may not arrive (corresponding to $x_i = 0$). In the divisible case, this variant is the \emph{prophet inequality for fractional knapsack}, recently studied by \cite{nuti2026sample} in the i.i.d.-drawn single-sample setting. In both cases, a $1/2$-competitive algorithm is known, and this ratio is tight.

For both cases, we will use $\SSSI$ to denote the single-item Game of Googol problem or an algorithm for that setting, and $\Opt(\{(q_i, x_i)\}_{i \in [n]})$ to denote the hindsight optimum. Similarly, we use $\kSSI$ to denote the single-item \emph{$k$-sample} Game of Googol problem or an algorithm for that setting. By a layer-cake representation,
\[\Opt(\{(q_i, x_i)\}_{i \in [n]}) = \int_0^\infty \min\left\{1, \sum_{i=1}^n x_i \cdot \bm{1}[q_i \ge w] \right\} \, dw.\]

% Reduction

\section{A General Reduction from Combinatorial to Single-Item PI} \label{sec:general-reduction}

In this section, we establish our general reduction from combinatorial allocation in the GoG setting to single-item allocation in the same setting. The reduction separates the two sources of feasibility in a combinatorial allocation problem: constraints on what may be allocated to each buyer, and supply constraints on each item.

When a buyer arrives, the algorithm first chooses a buyer-feasible \emph{provisional allocation} together with supporting prices. The provisional allocation enforces the buyer-side combinatorial constraint---for example, a unit-demand buyer is provisionally assigned at most one item---while the supporting prices assign additive value to the items and certify their value under the original valuation. This additivity allows the remaining problem to decompose across items.

The reduction then uses a single-item prophet inequality for each item to resolve conflicts among the buyers provisionally assigned that item. These per-item algorithms enforce the supply constraints and approximate the \emph{free-disposal value}, namely the maximum total supporting-price value that can be retained from the provisional allocation after resolving item conflicts. The free-disposal value therefore serves as the bridge between the two stages: the first stage seeks provisional allocations whose free-disposal value approximates the true combinatorial optimum, and the second converts that value into an ex-post feasible allocation.

At a high level, these two components interact online as follows. When buyer $i$ arrives:
\begin{enumerate}
\item The online free-disposal algorithm observes $v_i$ and produces a buyer-feasible provisional allocation $\vx_i$, together with a corresponding supporting-price vector $\vq_i$.
\item For each item $j\in[m]$, the $j$-th single-item prophet inequality receives the provisional allocation $x_{i,j}$ and its supporting price $q_{i,j}(\vx_i)$, and decides whether buyer $i$ receives item $j$ in the final allocation.
\end{enumerate}
%Thus, the first stage determines which items buyer $i$ may potentially receive and assigns additive values to them, while the second stage resolves competition for each item across buyers. In particular, 
If the free-disposal value $\FD(\vx_1,\ldots,\vx_n)$
is $\alpha$-competitive with $\Opt(\vv_1,\ldots,\vv_n)$, and, for every item $j$, the corresponding single-item prophet inequality is $\beta$-competitive with
$\Opt\left({(q_{i,j}(\vx_i),x_{i,j})}_{i\in[n]}\right)$,
then the two stages compose to give an $\alpha\beta$-competitive combinatorial prophet inequality.

There is, however, one important subtlety in this composition. To apply a single-item prophet inequality separately to each item, its inputs across buyers must retain the required independence. Consequently, the provisional allocation produced for buyer \(i\) cannot depend on the realized valuations of the other buyers. We capture this requirement through \emph{buyer-locality}. Formally, conditional on the shared sample, a provisional-allocation algorithm is buyer-local if the provisional allocation \(\vx_i\) and its supporting prices \(\vq_i\) depend only on buyer \(i\)'s valuation \(\vv_i\) (and any private randomness associated with buyer \(i\)), and not on \(\vv_{-i}\). Hence, conditional on the sample, the induced inputs to each single-item prophet inequality remain independent across buyers.

At first sight, buyer-locality makes it difficult to construct good provisional allocations: each buyer must be treated without knowing which items other buyers will value or provisionally demand, even though the quality of the free-disposal value depends on how these demands interact across buyers. This is precisely where a sample from the underlying distribution proves useful. The sample provides a common source of information that can be used to coordinate the buyer-local allocation rules across buyers. Thus, while the provisional allocation for buyer \(i\) cannot depend on the realized valuations of other buyers, it may depend on the shared sample, which summarizes the distributional environment and guides all buyers' provisional allocations in a consistent way. Conditional on this sample, the decisions remain buyer-local and preserve the independence needed by the single-item prophet inequalities.

%We now formalize this construction in the $k$-sample Game of Googol setting. 
The main result of this section is the following theorem for the $k$-sample Game of Googol setting, which we prove in \Cref{sec:reduction,sec:prov-alloc}. In \Cref{sec:2-to-1} we provide an improved reduction for the single-sample special case from the $\SSCPI$ to the $\SSSI$.

\begin{theorem} \label{cor:k-plus-1-sscpi-to-k-sspi}
    If $\FDAlg$ is an $\alpha$-competitive random-order free-disposal algorithm while $\kSSI$ is a $\beta_k$-competitive $k$-sample single-item allocation algorithm, then \Cref{alg:main-reduction} using \Cref{alg:ro-free-disp} as a subroutine provides an $(\alpha\beta_k/2)$-competitive $\kSCPI{(k + 1)}$.%-sample combinatorial allocation algorithm.
\end{theorem}

\subsection{Reducing \texorpdfstring{$\kSCPI{(k+1)}$}{(k+1)-SCombPI} to \texorpdfstring{$\kSSI$}{k-SPI}}
\label{sec:reduction}

We instantiate the general reduction in the $(k+1)$-Sample Combinatorial Prophet Inequality setting. Of the $k+1$ samples provided, we designate one as a {\em coordination sample} and denote it by $\vsc$. The remaining $k$ samples, denoted $\vs^1, \cdots, \vs^k$, are used for the underlying $\kSSI$s. The reduction uses two intermediate algorithms. The first, which we denote by $\PAlg$, is a single-sample provisional-allocation algorithm that uses the coordination sample to convert the buyers' valuations into a provisional allocation. The second, which we denote by $\kSSI$, is the $k$-sample prophet inequality. We run $m$ instantiations of this algorithm, one for each item. Each takes $k$ single-dimensional-valuation samples and item-specific provisional allocations and returns the final allocation for the buyer.

In order to preserve the symmetry of the samples and the real instance in the GoG setting of the $\kSSI$s, each of the $k$ samples $\vs^i$ is itself transformed first into a provisional allocation by $\PAlg$ using the same coordination sample $\vsc$.

The algorithm is formally described in pseudocode as \Cref{alg:main-reduction}. Before analyzing it, we define what it means for a provisional-allocation algorithm to be buyer-local.

\begin{definition}[Buyer-local provisional-allocation algorithm]
A provisional-allocation algorithm $\PAlg$ takes as input a coordination sample vector $\vsc$ and a real valuation vector $\vr$, and outputs, for every buyer $i\in[n]$, a provisional allocation $\vx_i$ together with a corresponding supporting-price vector $\vq_i$. We say that $\PAlg$ is \emph{buyer-local} if, conditioned on $\vsc$ and on any randomness shared across buyers, the pair $(\vx_i,\vq_i)$ depends only on $r_i$ and on randomness local to buyer $i$.

Equivalently, after fixing $\vsc$ and the shared randomness, $\PAlg$ can be implemented by algorithms $\PAlg_1,\ldots,\PAlg_n$, where $\PAlg_i$ receives only $r_i$ and produces $(\vx_i,\vq_i)$. In particular, $\PAlg$ admits an online implementation when the valuations $r_1,\ldots,r_n$ arrive sequentially.
\end{definition}

\suppressfloats[t]
\begin{algorithm}[!htb]
    \caption{Combinatorial to Single-Item Prophet Inequalities Reduction}
    \label{alg:main-reduction}
    \begin{algorithmic}[1]
    \Input{
        \item[]
        \begin{itemize}
            \item Sample valuation vectors $(\vsc, \vs^1, \dots, \vs^k)$, where $\vsc$ is the coordination sample,
            \item Arriving buyers with valuations $\vr$,
            \item A buyer-local provisional-allocation algorithm $\PAlg$, and
            \item A $k$-sample single-item Game of Googol algorithm $\kSSI$ for the appropriate setting $A$.
        \end{itemize}}
    \Output{Feasible allocation vectors $\vx'_1, \dots, \vx'_n$.}
    \State Initialize $m$ independent instances of $\kSSI$ as $\kSSI_1, \dots, \kSSI_m$.
    \State Fix any internal randomness of $\PAlg$ and use the same coin flips for all of its oracle calls.
    \For {sample $t \in [k]$}
        \State Let $\vx^t_1, \dots, \vx^t_n \gets \PAlg(\vsc, \vs^t)$.
        \State Let $\vq^t_1, \dots, \vq^t_n$ be the supporting prices of the allocation vectors, i.e., $\vq^t_i \gets \vq_{s^t_i}(\vx^t_i)$.
        \LineComment{Assemble the $t$-th sample for the $\kSSI$s.}
        \For {item $j \in [m]$}
            \State Extract $(q^t_{i, j}, x^t_{i, j})$ for each buyer $i \in [n]$ from the above.
            \State Let $(q^t_{1, j}, x^t_{1, j}), (q^t_{2, j}, x^t_{2, j}), \dots, (q^t_{n, j}, x^t_{n, j})$ be the $t$-th sample for $\kSSI_j$.
        \EndFor
    \EndFor
    \LineComment{Extract the per-item allocation from each $\kSSI$.}
    \For {buyer $i \in [n]$} 
        \State $\widehat{\vx}_i \gets \PAlg_i(\vsc, r_i)$. \Comment{$\PAlg$ needs to be buyer-local for this to be well-defined.}
        \State Let $\widehat{\vq}_i \gets \vq_{r_i}(\widehat{\vx}_i)$ be the supporting price of this allocation vector.
        \For {item $j \in [m]$}
            \State Send $(\widehat{q}_{i, j}, \widehat{x}_{i, j})$ as the real valuation of buyer $i$ to $\kSSI_j$, and receive an allocation $x'_{i, j} \le \widehat{x}_{i, j}$.
        \EndFor
        \State Aggregate the per-item allocations into $\vx'_i$ and allocate it to buyer $i$.
    \EndFor
    \end{algorithmic}
\end{algorithm}

%The first step in our general framework is a reduction from the $(k + 1)$-sample combinatorial Game of Googol problem ($\kSCPI{(k+1)}$) to the $k$-sample single-item Game of Googol problem ($\kSSI$). Since the valuation functions have supporting prices, we can decompose the combinatorial problem into $m$ instances of $\kSSI$ using these prices. The existence of a combinatorial allocation rule with a good free-disposal value then immediately implies the existence of a combinatorial prophet inequality.

%In particular, assume that we have access to a \emph{buyer-local allocation algorithm} $\Alg$, which takes a coordination sample $\vsc$ and a real valuation vector $\vr$ and outputs allocation vectors $\vx_1, \dots, \vx_n$ for the real valuation vector. We now define what it means to be buyer-local.

We measure the performance of a buyer-local algorithm through the free-disposal value in the Game of Googol model.

\begin{definition}[Free-disposal competitiveness]
A buyer-local provisional-allocation algorithm $\PAlg$ is \emph{$\alpha$-competitive under the free-disposal value} if the following holds for every collection of valuation pairs $\{v_{i,1},v_{i,2}\}_{i\in[n]}$. Independently for each buyer $i$, let $(s_i^c,r_i)$ be a uniformly random ordering of $(v_{i,1},v_{i,2})$. Then
    \[\E_{\vsc, \vr}[\FD(\PAlg(\vsc, \vr))] \ge \alpha \cdot \E_{\vr}[\Opt(\vr)].\]
    %\[\E_{\vsc, \vr}[\FD(\vx_1, \dots, \vx_n)] \ge \alpha \cdot \E_{\vr}[\Opt(\vr)].\]
where the expectation is over the independent random orderings and any internal randomness of $\PAlg$.
\end{definition}

We observe that \Cref{alg:main-reduction} %uses $k + 1$ samples in total, is well-defined, and 
always outputs a feasible allocation, since for all $j$ we have $\sum_{i=1}^n x'_{i, j} \le 1$ by the guarantee of $\kSSI_j$. We show that if $\PAlg$ produces allocation vectors with a good free-disposal value in the Game of Googol setting, then \Cref{alg:main-reduction} is a good $(k + 1)$-sample combinatorial prophet inequality. The proof is a straightforward application of the per-item decomposition property of the free-disposal value, and we defer it to \Cref{app:general-reduction}.

\begin{lemma} \label{lem:combi-pi-to-pi}
    Suppose that $\PAlg$ is buyer-local and $\alpha$-competitive under the free-disposal value, while $\kSSI$ is $\beta$-competitive. Then, \Cref{alg:main-reduction} is $\alpha\beta$-competitive.
\end{lemma}

\subsection{Constructing a Competitive Buyer-Local Provisional-Allocation Algorithm}
%Reducing to a Random-Order Online Free Disposal Algorithm}
\label{sec:prov-alloc}

We are now tasked with designing a buyer-local free-disposal algorithm $\PAlg$ in the single-sample GoG setting. Let $\vs$ denote the sample vector and $\vr$ the real valuation vector. Our approach is to transform a worst-case competitive online free-disposal algorithm into a buyer-local algorithm, using the coordination sample to exploit the symmetry between sample and real valuations.

Let $\FDAlg$ be an online free-disposal algorithm such that, on every input instance $\vv$, the free-disposal value of its provisional allocation is competitive with $\Opt(\vv)$. We construct $\PAlg$ from $n$ coupled executions of $\FDAlg$, one for each buyer. In the execution associated with buyer $i$, we run $\FDAlg$ on the hybrid instance obtained from $\vs$ by replacing $s_i$ with the real valuation $r_i$; all $n$ executions use the same internal randomness. The provisional allocation (and supporting prices) assigned to buyer $i$ in this execution are returned by $\PAlg_i$. By construction, $\PAlg_i$ depends on the real instance only through $r_i$, and hence $\PAlg$ is buyer-local.

Buyer-locality also separates the real arrival order from the order used within the simulated executions of $\FDAlg$. The components $\PAlg_i$ may be invoked in whatever order the real buyers arrive, while the hybrid instance supplied to $\FDAlg$ may be processed according to any ordering chosen by the reduction. We can therefore choose this internal ordering to obtain the strongest available guarantee for $\FDAlg$. In particular, since $\FDAlg$ enjoys a better competitive ratio under random arrival order than under adversarial order, we process the hybrid instances in a uniformly random order and inherit this improved guarantee in the reduction.

%take a \emph{random-order online free disposal algorithm} and represent each online stage by an algorithm $\Alg_i$. In particular, given a random-order online free disposal algorithm $\FDAlg$, we first run $\FDAlg$ on the samples $s_1, \dots, s_n$ in a uniformly random order $\pi$. Then, for every buyer $i$, we define $\Alg_i$ to be the algorithm that returns the allocation the buyer would have received had we used $r_i$ instead of $s_i$ in the sample run. This ensures the local property of $\Alg$, since the rule for allocating $\vx_i$ depends only on $r_i$ and the fictitious run on the sample $\vs$.

%To formalize this idea, we let $\FDAlg$ be an algorithm that takes two inputs: the state $\Sigma$ of the algorithm before the current buyer arrives and the buyer's valuation function $v$. The state of the algorithm may contain the preceding buyers and their allocations, as well as a source of random bits to be used by the algorithm. The algorithm then produces an allocation vector $\vx \in A^m$ and the state $\Sigma'$ of the algorithm after the current buyer arrives. Note that since the state $\Sigma$ contains the source of random bits used to produce $\vx$, the output of $\FDAlg$ is deterministic conditional on $\Sigma$ and $v$.

%Given this setup, we let $\Alg_i$ be exactly $\FDAlg(\Sigma_i, r_i)$, where $\Sigma_i$ is the state of the algorithm before $s_i$ arrives. Due to the randomness of $s_i$ and $r_i$, we can prove that the free-disposal value of $(\Alg_1, \dots, \Alg_n)$ is at least half that of $\FDAlg$ on $\vs$.

\suppressfloats[t]
\begin{algorithm}[!htb]
    \caption{Producing a Buyer-Local Free-Disposal Algorithm}
    \label{alg:ro-free-disp}
    \begin{algorithmic}[1]
    \Input{\item[]
    \begin{itemize}
        \item Sample valuations $\vs$,
        \item Arriving buyers with valuations $\vr$, and
        \item A random-order free-disposal algorithm $\FDAlg$.
    \end{itemize}}
    \Output{A buyer-local algorithm $\PAlg$.}
    \State Let $\pi$ be a random permutation of $\{1, \dots, n\}$, and let $\Sigma^{(0)}$ be the initialization state of $\FDAlg$ including all internal coin flips.
    %, with a fixed source of randomness if $\FDAlg$ is a randomized algorithm.
    \For {$t$ from $1$ to $n$}
        \State Let $i \gets \pi_t$ be the current buyer index.
        \State Record $\Sigma_i \gets \Sigma^{(t - 1)}$ as the state of the algorithm before $s_i$ is processed.
        \State $(\vz_i, \Sigma^{(t)}) \gets \FDAlg(\Sigma^{(t - 1)}, s_i)$. \Comment{$\vz_i$ is used in analysis only}
    \EndFor
    \State For every buyer $i \in [n]$, $\PAlg_i(\vs, r_i)\gets \FDAlg(\Sigma_i, r_i)$. 
    %let $\PAlg_i$ be such that the output of $\PAlg_i(\vs, r_i)$ is the allocation from $\FDAlg(\Sigma_i, r_i)$.
    \State \Return $\PAlg \gets (\PAlg_1, \dots, \PAlg_n)$.
    \end{algorithmic}
\end{algorithm}

The reduction from buyer-local provisional allocation to online free disposal is formalized in Algorithm~\ref{alg:ro-free-disp}.
We now show that $\PAlg$ inherits the competitive ratio of $\FDAlg$ within a factor of $2$.

\begin{lemma} \label{lem:fdalg-to-ss}
    Suppose $\FDAlg$ is any random-order free-disposal algorithm. Fix $(v_{i, 1}, v_{i, 2})$ for all buyers $i$, and suppose $\vs$ and $\vr$ are sampled according to the Game of Googol from these fixed valuations. Then $\PAlg$ from \Cref{alg:ro-free-disp} is a buyer-local algorithm, and
    \[\E[\FD(\PAlg(\vs, \vr))] \ge \frac{1}{2} \cdot \E[\FD(\FDAlg(\vr))].\]
    %where $\vx'_i = \Alg_i(\vs, r_i)$.
\end{lemma}

\begin{proof}
    $\PAlg$ is buyer-local by construction, so we will only prove the second claim. Fix the random arrival order $\pi$. By relabeling the buyers, we assume that $\pi$ is the identity permutation. Let $\vz_i=\FDAlg(\Sigma_i, s_i)$ denote the allocation made by $\FDAlg$ on $\vs$ to agent $i$, and $\vx_i=\FDAlg(\Sigma_i, r_i)$ the provisional allocation made to $i$ by $\PAlg_i(\vs, r_i)$. 
    
    For some agent $k\in [n]$, consider the free-disposal value obtained by the partial allocation $\vx_1, \dots, \vx_k$:
    \[
        \FD(\vx_1, \dots, \vx_k) = \sum_{j=1}^m \int_0^\infty \underbrace{\min\left\{1, \sum_{i=1}^k x_{i, j} \cdot \bm{1}[q_{i, j}(\vx_i) \ge w] \right\}}_{y_{j, w}(\vx_{1..k})} \,dw.
    \]
    
    To prove the lemma, it suffices to show that for every value $w \ge 0$ and every initial state $\Sigma^{(0)}$, at the end of the algorithm we have
    \[
        \E[y_{j, w}(\vx_{1..n})] \ge \frac 12 \cdot \E[y_{j, w}(\vz_{1..n})] ,
    \]
    where the expectation is taken over the random assignment of $v_{i,1}$ and $v_{i,2}$ to $s_i$ and $r_i$ for all $i\in [n]$. Then integrating over $w$ and summing over $j$ provides 
    $\E[\FD(\PAlg(\vs, \vr))] \ge \frac{1}{2} \cdot \E[\FD(\FDAlg(\vs))]$,
    and the lemma follows by the symmetry of $\vs$ and $\vr$.

    We prove the claim by showing
    \[
        \E[y_{j, w}(\vx_{1..n})] \ge \frac{1}{2} \cdot \E[y_{j, w}(\vx_{1..n} \cup \vz_{1..n})] \ge \frac 12 \cdot \E[y_{j, w}(\vz_{1..n})],
    \]
    where $y_{j, w}(\vx_{1..n} \cup \vz_{1..n}) := \min\left\{1, \sum_{i=1}^n \left(x_{i, j} \cdot \bm{1}[q_{i, j}(\vx_i) \ge w] + z_{i, j} \cdot \bm{1}[q_{i, j}(\vz_i) \ge w]\right) \right\}$ can be seen as the ``merged'' free-disposal contribution of item $j$ at weight level $w$.

    Fix any partial realization of $s_1, \dots, s_{k - 1}$ for some $k \in [n]$, which also fixes $\Sigma^{(k - 1)}$ to be the state of $\FDAlg$ before the $k$-th buyer. We now show that
    \begin{equation} \label{eq:telescoper}
        \E[y_{j, w}(\vx_{1..k}) - y_{j, w}(\vx_{1..k - 1}) \mid \vs_{1..k-1}] \ge \frac{1}{2} \left(y_{j, w}(\vx_{1..k} \cup \vz_{1..k}) - y_{j, w}(\vx_{1..k-1} \cup \vz_{1..k - 1})\right)
    \end{equation}
    where the RHS is deterministic conditional on $s_1, \dots, s_{k - 1}$ because it does not depend on the random assignment of $s_i$ and $r_i$ to the two $i$-th valuations.
    
    The claim follows immediately from the concavity of $f(x) = \min\{1, x\}$. Let
    \begin{align*}
        A &= \sum_{i=1}^{k - 1} x_{i, j} \cdot \bm{1}[q_{i, j}(\vx_i) \ge w], \\
        B &= \sum_{i=1}^{k - 1} \left(x_{i, j} \cdot \bm{1}[q_{i, j}(\vx_i) \ge w] + z_{i, j} \cdot \bm{1}[q_{i, j}(\vz_i) \ge w]\right),\\
        c &= u_{k, j} \cdot \bm{1}[q_{k, j}(\vu_k) \ge w] \quad\quad\quad\quad \text{where $\vu_k = \FDAlg(\Sigma^{(k - 1)}, v_{k, 1})$}, \\
        d &= u'_{k, j} \cdot \bm{1}[q_{k, j}(\vu'_k) \ge w] \quad\quad\quad\quad \text{where $\vu'_k = \FDAlg(\Sigma^{(k - 1)}, v_{k, 2})$},
    \end{align*}
    then $0 \le A \le B$ and $c, d \ge 0$. Furthermore,
    \[\E[y_{j, w}(\vx_{1..k}) - y_{j, w}(\vx_{1..k - 1}) \mid \vs_{1..k-1}] = \frac{1}{2} \left(f(A + c) - f(A) + f(A + d) - f(A)\right)\]
    while
    \[y_{j, w}(\vx_{1..k} \cup \vz_{1..k}) - y_{j, w}(\vx_{1..k-1} \cup \vz_{1..k - 1}) = f(B + c + d) - f(B).\]
    Finally, by the concavity of $f(x)$, $f(B + c + d) - f(B) = f(B + c + d) - f(B + d) + f(B + d) - f(B) \le f(A + c) - f(A) + f(A + d) - f(A)$. This completes the proof of \Cref{eq:telescoper}.
%    By summing \Cref{eq:telescoper} over $k$ from $1$ to $n$ and telescoping, and using the fact that $\pi$ is chosen independently of the internal randomness of $\FDAlg$, we obtain the lemma.
\end{proof}

We then have the following corollary.

\begin{corollary} \label{cor:random-order-fd-to-tuple}
    If $\FDAlg$ is an $\alpha$-competitive random-order free-disposal algorithm, then $\PAlg$ as constructed above is an $\alpha/2$-competitive buyer-local algorithm.
\end{corollary}

By combining with \Cref{lem:combi-pi-to-pi}, we obtain \Cref{cor:k-plus-1-sscpi-to-k-sspi}.
%then obtain the following end-to-end reduction for $(k+1)$-sample combinatorial prophet inequality.

%We remark that \Cref{alg:ro-free-disp} can be modified by taking $\pi$ to be a fixed order and $\FDAlg$ to be an adversarial-order free disposal algorithm; \Cref{cor:random-order-fd-to-tuple,cor:k-plus-1-sscpi-to-k-sspi} would still follow from the competitive ratio of $\FDAlg$. However, as we will see in \Cref{sec:free-disp}, we can obtain strictly better performance by harnessing the power of random-order algorithms on the samples.

\subsection{From Two Samples to One Sample}
\label{sec:2-to-1}

Finally, we turn \Cref{alg:main-reduction} with $k = 1$, which uses two samples $(\vsc, \vs^1)$, into an algorithm that uses only one sample $\vs'$. Such a reduction is already known from \cite{dutting2024online}, which proves that any $\kSCPI{k}$ algorithm can be turned into a $\SSCPI$ algorithm at a loss factor of $k + 1$. For \Cref{alg:main-reduction}, this result translates to a loss factor of $3$ in the competitive ratio, which amounts to an $(\alpha \beta/6)$-competitive $\SSCPI$ algorithm given an $\alpha$-competitive free-disposal algorithm and a $\beta$-competitive $\SSSI$ algorithm. We show that for our specific framework, we can reduce the incurred loss to a factor of $2$ or less.

The main idea here is to split the buyers into two types: the \emph{coordinating} and the \emph{competition} buyers. In particular, suppose each buyer $i \in [n]$ has a sample valuation $s'_i$ and real valuation $r'_i$. Then we construct the two-sample instance $(\vsc, \vs^1, \vr)$ to feed into \Cref{alg:main-reduction} as follows. Fix a parameter $q\in [0,1]$ and for each buyer $i$ independently:
\begin{enumerate}
    \item With probability $q$, buyer $i$ plays the role of a coordinating buyer. In particular, $(s^c_i, s^1_i, r_i) \gets (s'_i, 0, 0)$.
    \item With probability $1-q$, buyer $i$ is a competition buyer. These buyers directly compete in the free-disposal algorithm. In particular, $(s^c_i, s^1_i, r_i) \gets (0, s'_i, r'_i)$.
\end{enumerate}

%To determine the role of each buyer, we flip a biased coin identically and independently for each buyer, so that a buyer is coordinating with some parameterized probability $q \in [0, 1]$. 
\noindent
For completeness, we provide the pseudocode of the algorithm in \Cref{alg:sscpi}.

\suppressfloats[t]
\begin{algorithm}[!htb]
    \caption{Single-Sample Combinatorial Online Allocation}
    \label{alg:sscpi}
    \begin{algorithmic}[1]
    \Input{
        \item[]
        \begin{itemize}
            \item Sample valuations $\vs'$ and arriving buyers with valuations $\vr'$,
            \item A random-order free-disposal allocation algorithm $\FDAlg$,
            \item A single-sample single-item Game of Googol algorithm $\SSSI$ for the appropriate setting $A$, and
            \item A parameter $q \in [0, 1]$.
        \end{itemize}}
    \Output{Feasible allocation vectors $\vx_1, \dots, \vx_n$.}
    \State Independently draw $H_i \sim \operatorname{Bernoulli}(q)$ for every buyer $i \in [n]$, where $H_i = 1$ means that buyer $i$ is coordinating and $H_i = 0$ means that buyer $i$ is competing.
    \For {buyer $i \in [n]$}
        \IfThenElse{$H_i = 1$}{$(s^c_i, s^1_i) \gets (s'_i, 0)$}{$(s^c_i, s^1_i) \gets (0, s'_i)$}.
    \EndFor
    \State Construct the buyer-local algorithm $\PAlg$ from $\vsc$ and $\FDAlg$ as in \Cref{alg:ro-free-disp}.
    \State Initialize \Cref{alg:main-reduction} with $k = 1$, sample valuations $(\vsc, \vs^1)$, buyer-local algorithm $\PAlg$, and single-item algorithm $\SSSI$.
    \For {each arriving buyer $i$ with valuation $r'_i$}
        \State Let $r_i \gets 0$ if $H_i = 1$, and let $r_i \gets r'_i$ otherwise.
        \State Send $r_i$ to \Cref{alg:main-reduction} and receive its allocation $\widehat{\vx}_i$.
        \IfThenElse{$H_i = 1$}{allocate $\vx_i \gets \bm{0}$}{allocate $\vx_i \gets \widehat{\vx}_i$}.
        \State Let $\vz_i \gets \PAlg_i(\vsc, s^c_i)$. In other words, $\vz_i$ is exactly the allocation produced by $\FDAlg$ running on $\vsc$. \Comment{Only for analysis}
    \EndFor
    \end{algorithmic}
\end{algorithm}

We now prove a guarantee for \Cref{alg:sscpi} assuming black-box guarantees for $\FDAlg$ and $\SSSI$. For $\SSSI$, we simply assume that the algorithm is $\beta$-competitive. For $\FDAlg$, we assume that it has a \emph{smooth competitive ratio} that depends on the fraction of buyers that get zeroed out. Formally,

\begin{definition}[Smooth competitive ratio] \label{def:smooth-competitive}
    Fix any (random-order) online free-disposal allocation algorithm $\FDAlg$. Let $\Gamma : [0, 1] \to [0, 1]$ be any increasing function. Then, we say $\FDAlg$ has a \emph{smooth competitive ratio} $\Gamma$ if for any valuation vector $\vv$ on $n$ buyers, and for any $q \in [0, 1]$, the following is true.
    \begin{quote}
        Randomly construct the valuation vector $\vv'$ as follows: for every buyer $i$ independently, $v'_i = v_i$ with probability $q$ and $v'_i = 0$ otherwise. Run $\FDAlg$ on $\vv'$ and obtain allocation vectors $\vx_1, \dots, \vx_n$. Then,
        \[\E[\FD(\vx_1, \dots, \vx_n)] \ge \Gamma(q) \cdot \Opt(v_1, \dots, v_n)\]
        where the left-hand side is the expected free-disposal value of $\vx_1, \dots, \vx_n$ evaluated on $\vv'$, and the expectation is taken over $\vv'$ and the internal randomness of $\FDAlg$.
    \end{quote}

     In particular, $\Gamma(1)$ is a competitive guarantee of $\FDAlg$. For the random-order case, the following viewpoint is equivalent. Assume that the algorithm simply skips over zero valuations, and let $n' \sim \mathrm{Binom}(n, q)$. Then the expected free-disposal value from the first $n'$ buyers of the random order $\pi$ is at least $\Gamma(q)$ of the hindsight optimum.
\end{definition}

Fix the valuations $v_{i, 1}, v_{i, 2}$ for all buyers $i$, and assume $\vs'$ and $\vr'$ are drawn in a Game of Googol fashion from these valuations. We first show the following technical lemma.

\begin{lemma}\label{lem:biased-mirror}
For every $q \in [0,1]$, $\E[\FD(\vx_1,\dots,\vx_n)] \ge (1-q)\E[\FD(\vz_1,\dots,\vz_n)]$.
\end{lemma}

This inequality is the analogue of \Cref{lem:fdalg-to-ss} promised above. Indeed, when $q=1/2$, its instancewise inequality is the same real-versus-sample replay comparison as \Cref{lem:fdalg-to-ss}, where the replayed allocations $\vx_1, \dots, \vx_n$ have at least half the free-disposal value of the full sample allocations $\vz_1, \dots, \vz_n$. Here the real side consists of the competing buyers and the sample side consists of the coordinating buyers; \Cref{lem:biased-mirror} extends the factor $1/2$ to the biased factor $1-q$. We defer its proof to \Cref{app:general-reduction}.

Noting that $\E[\FD(\vy_1,\dots,\vy_n)] \ge \Gamma(q) \cdot \E[\Opt(s'_1, \dots, s'_n)]$ by the smooth competitive ratio definition of $\FDAlg$, we get the following guarantee for \Cref{alg:sscpi}.

\begin{theorem}\label{thm:sscpi}
    If $\FDAlg$ has a smooth competitive ratio of $\Gamma$, while $\SSSI$ is $\beta$-competitive, then with $q = \argmax (1 - q)\Gamma(q)$, \Cref{alg:sscpi} is a $\left(\beta \cdot \max_{q \in [0, 1]} (1 - q)\Gamma(q)\right)$-competitive $\SSCPI$ algorithm.
\end{theorem}

\begin{proof}
    Fix any $q \in [0, 1]$. By the smooth competitive ratio of $\FDAlg$ combined with \Cref{lem:biased-mirror}, we have
    \begin{align*}
    \E[\FD(\vx_1,\dots,\vx_n)] &\ge (1-q)\Gamma(q) \cdot \E[\Opt(s'_1, \dots, s'_n)] \\
    &= (1-q)\Gamma(q) \cdot \E[\Opt(r'_1, \dots, r'_n)].
    \end{align*}

    To apply the guarantee of $\SSSI$, it remains to show that $s^1_i$ and $r_i$ are sampled in a Game of Googol fashion conditioned on the coin flips $H_1, \dots, H_n$; if this is true, then the rest of the proof follows similarly to the proof of \Cref{lem:combi-pi-to-pi}. This is indeed true: conditioned on the role coin $H_i$, either $(s^1_i, r_i)$ is shuffled from $(0, 0)$, or $(v_{i, 1}, v_{i, 2})$.
\end{proof}

As we can naively bound $\Gamma(q) \ge \alpha q$ where $\alpha$ is the competitive ratio of $\FDAlg$, we immediately get the following corollary.

\begin{corollary} \label{lem:2-sample-to-1-sample}
If $\FDAlg$ is an $\alpha$-competitive random-order free-disposal algorithm and $\SSSI$ is $\beta$-competitive, then with $q = 1/2$, \Cref{alg:sscpi} is a $(\alpha\beta/4)$-competitive $\SSCPI$ algorithm.
\end{corollary}

\subsection{Computational Considerations}

Given oracle access to the $\FDAlg$, $\kSSI$, and $\SSSI$ algorithms, Algorithms~\ref{alg:main-reduction}, \ref{alg:ro-free-disp}, and \ref{alg:sscpi} run in polynomial time. Therefore, we immediately obtain efficient $\kSCPI{(k+1)}$ and $\SSCPI$ algorithms, assuming that the underlying $\FDAlg$, $\kSSI$, and $\SSSI$ algorithms can be implemented efficiently.

% Random-order Free Disposal Algorithms
\section{Random-Order Free-Disposal Algorithms} \label{sec:free-disp}

In this section, we provide random-order free-disposal algorithms for several classes of valuations. In particular, for the indivisible setting, we show that the greedy algorithm \cite{dobzinski2005approximation} obtains a competitive ratio strictly greater than $1/2$ under random buyer arrivals. Notably, we show that $1/2$ is the best competitive ratio for efficient free-disposal algorithms under \emph{adversarial} XOS-buyer arrivals. This result demonstrates the ability of our framework to leverage random order in single-sample prophet inequalities. We also provide several first results on free-disposal algorithms for the divisible case.

\subsection{Indivisible-Item Free-Disposal Algorithm}

Consider the following buyer-wise greedy algorithm (\Cref{alg:bwg}), first studied by \cite{dobzinski2005approximation} in the context of welfare maximization for XOS buyers and then used extensively in \cite{dutting2024online} for the first single-sample combinatorial prophet inequality for XOS buyers. Here, for a set $S \subseteq [m]$, we use the notation $\vp(S)$ to denote $\sum_{j \in S} p_j$.

\suppressfloats[t]
\begin{algorithm}[!htb]
    \caption{Buyer-Wise Greedy \cite{dobzinski2005approximation,dutting2024online}}
    \label{alg:bwg}
    \begin{algorithmic}[1]
    \Input{Arriving buyers with valuations $\vv$.}
    \Output{Allocation subsets $S_1, \dots, S_n$ that need not be feasible.}
    \State Initialize $\vp^{(0)} \gets \bm{0}$.
    \For {buyer $i \in [n]$}
        \State $S_i \gets \argmax_{S \subseteq [m]} \left(v_i(S) - \vp^{(i - 1)}(S)\right)$.
        \For {item $j \in [m]$}
            \State $p_j^{(i)} \gets q_{i, j}(S_i)$ if $j \in S_i$, else $p_j^{(i)} \gets p_j^{(i - 1)}$.
        \EndFor
        \State Allocate $S_i$ to buyer $i$.
    \EndFor
    \end{algorithmic}
\end{algorithm}

We prove that \Cref{alg:bwg} achieves a competitive ratio of $0.5096$ and a smooth competitive ratio of $\Gamma(q) = q - q^2/2$, under the free-disposal value in the random-order buyer-arrival setting. The competitive ratio is a strict improvement over the ratio of $1/2$ achieved by the same algorithm under an adversarial arrival order \cite{dutting2024online,dobzinski2005approximation}. In fact, we also show that no efficient algorithm has a competitive ratio greater than $1/2$ under the free-disposal value for an adversarial arrival order.

\begin{theorem} \label{thm:bwg-random-order}
    For buyers with XOS valuations that arrive in a uniformly random order, \Cref{alg:bwg} is $0.5096$-competitive and $\Gamma$-smooth-competitive under the free-disposal value, where $\Gamma(q) = q - q^2/2$. Furthermore, the algorithm runs in polynomial time and uses polynomially many demand and supporting-price queries.
\end{theorem}

\begin{theorem} \label{thm:lb-xos-adv-order}
    For buyers with XOS valuations that arrive in an adversarial order, no efficient algorithm can achieve a competitive ratio of $1/2 + \eps$ under the free-disposal value for any constant $\eps > 0$, unless $\textnormal{\textsc{NP}} = \textnormal{\textsc{RP}}$. Here, an algorithm is efficient if its running time is polynomial in $n$, $m$, and the description complexity of the valuations.
\end{theorem}

We defer the proof of the hardness result to \Cref{app:free-disp} and focus on the algorithmic side in this section. Our proof is a reduction to online monotone submodular maximization over a simple partition matroid. This problem was studied by \cite{buchbinder2020online} as a generalization of online monotone submodular maximization. We use the following result as a black box.

\begin{theorem} \label{thm:random-order-partition-greedy}
    Let $\mathcal{N}$ be a finite ground set partitioned into nonempty, pairwise disjoint parts $P_1,\dots,P_\ell$, and let $f:2^{\mathcal{N}} \to \R_{\ge 0}$ be a nonnegative monotone submodular function. A set $S\subseteq\mathcal{N}$ is feasible if $|S\cap P_i|\le 1$ for every $i\in[\ell]$, and the objective is to maximize $f(S)$ over all such sets. The parts are revealed one at a time in an order given by a uniformly random permutation $\pi$ of $[\ell]$. Upon seeing $P_{\pi(t)}$, the algorithm must immediately and irrevocably choose an element from this part.

    Starting with $A_0 \gets \emptyset$, consider the greedy choices
    \[
        u_t \gets \argmax_{u\in P_{\pi(t)}} f(u\mid A_{t-1}),
        \qquad A_t \gets A_{t-1}\cup\{u_t\},
    \]
    where $f(u\mid A):=f(A\cup\{u\})-f(A)$ and ties are broken arbitrarily.
    Then
    \[
        \E_{\pi}\bigl[f(A_\ell)\bigr]
        \ge 0.5096\cdot
        \max_{\substack{S\subseteq\mathcal{N}\\
                        |S\cap P_i|\le 1}} f(S).
    \]
    In other words, random-order greedy is $0.5096$-competitive for maximizing $f$ subject to the simple partition matroid constraint \citep[Theorem 1]{buchbinder2020online}. Furthermore, random-order greedy is $\Gamma$-smooth-competitive as well, where $\Gamma(q) = q - q^2/2$ \citep[essentially Lemma 3.4]{buchbinder2020online}.
\end{theorem}

In \cite{buchbinder2020online}, the $0.5096$ competitive ratio is stated as is, while the smooth competitive ratio is stated for the free-disposal value of exactly the first $q$ fraction of arriving buyers, which is not exactly our definition of smooth competitive ratio. We therefore prove the independently pruned version here in \Cref{lem:partition-greedy-smooth}.

We begin the reduction. For every buyer $i$ and bundle $S\subseteq[m]$, let $e_{i,S}=(i,S)$ denote a configuration. Let
\[
    \mathcal{N}=\{e_{i,S}:i\in[n],\ S\subseteq[m]\}
\]
and
\[
    P_i=\{e_{i,S}:S\subseteq[m]\}.
\]
The sets $P_1,\dots,P_n$ define a simple partition matroid: a feasible partial allocation contains at most one element from each $P_i$. For $\mathcal{I}\subseteq\mathcal{N}$, define
\[
    f(\mathcal{I})
    :=\sum_{j=1}^m
    \max_{e_{i, S} \in \mathcal{I}, S \ni j} q_{i, j}(S).
\]
Observe that the partial allocations $\{S_i : i \in I \subseteq [n]\}$ have a one-to-one mapping to a feasible subset $\mathcal{I} = \{e_{i, S_i} : i \in I \subseteq [n]\}$ of $\mathcal{N}$. Furthermore, $f(\mathcal{I})$ is exactly the free-disposal value of the encoded partial allocation. Therefore, finding good free-disposal allocations is equivalent to finding a feasible set with a good $f$ value.

We first prove that $f$ is submodular.

\begin{lemma} \label{lem:configuration-fd-submodular}
    The function $f$ is normalized, nonnegative, monotone, and submodular.
\end{lemma}

\begin{proof}
    We prove submodularity here. For every item $j$, define
    \[
        p_j(\mathcal{I})
        :=\max_{e_{i, S} \in \mathcal{I}, S \ni j} q_{i, j}(S)
    \]
    to be the contribution of item $j$ to $f$. If $\mathcal{I}\subseteq\mathcal{J}$, then $p_j(\mathcal{I})\le p_j(\mathcal{J})$ for every $j$. Consequently, the marginal contribution of $e_{i,S}$ is
    \[
        f(e_{i,S}\mid\mathcal{I})
        =\sum_{j\in S}\left(q_{i,j}(S)-p_j(\mathcal{I})\right)^+
        \ge
        \sum_{j\in S}\left(q_{i,j}(S)-p_j(\mathcal{J})\right)^+
        =f(e_{i,S}\mid\mathcal{J}),
    \]
    which is the diminishing-returns definition of submodularity.
\end{proof}

It remains to show that the greedy algorithm for the configuration instance behaves exactly like \Cref{alg:bwg} on the original instance. Suppose that buyers $1,\dots,i-1$ have been processed. Let $\mathcal{I}_{i-1}$ contain the encoded partial allocation, and let $p_j=p_j(\mathcal{I}_{i-1})$. For a candidate bundle $S$, define its buyer-wise utility and marginal contribution by
\begin{align*}
    U_i(S)&:=v_i(S)-\vp(S)
        =\sum_{j\in S}\left(q_{i,j}(S)-p_j\right),\\
    \Delta_i(S)&:=f(e_{i,S}\mid\mathcal{I}_{i-1})
        =\sum_{j\in S}\left(q_{i,j}(S)-p_j\right)^+.
\end{align*}

The following lemma shows that every maximizer of $U_i$ is also a maximizer of $\Delta_i$. Hence, \Cref{alg:bwg} is precisely greedy on the configuration instance.

\begin{lemma} \label{lem:bwg-is-configuration-greedy}
    We have $\max_{S\subseteq[m]}U_i(S)=\max_{S\subseteq[m]}\Delta_i(S)$. Moreover, every maximizer of $U_i$ is also a maximizer of $\Delta_i$.
\end{lemma}

\begin{proof}
    Clearly, $U_i(S)\le\Delta_i(S)$. For the reverse direction, fix $S$ and let $S'=\{j\in S:q_{i,j}(S)>p_j\}$. Applying the supporting-price property to $S' \subseteq S$ gives
    \[
        U_i(S')=v_i(S')-\vp(S')
        \ge\sum_{j\in S'}\left(q_{i,j}(S)-p_j\right)
        =\Delta_i(S).
    \]
    Taking the maximum on both sides gives $\max_S U_i(S)\ge\max_S\Delta_i(S)$, while the pointwise inequality gives the reverse inequality. If $S_i$ maximizes $U_i$, then
    \[
        U_i(S_i)\le\Delta_i(S_i)\le\max_S\Delta_i(S)
        =\max_S U_i(S)=U_i(S_i),
    \]
    so all inequalities are equalities and $S_i$ also maximizes the marginal contribution.
\end{proof}

\begin{proof}[Proof of \Cref{thm:bwg-random-order}]
    A uniformly random buyer order is a uniformly random order of the parts $P_i$. By \Cref{lem:configuration-fd-submodular}, the configuration objective is monotone submodular; by \Cref{lem:bwg-is-configuration-greedy}, \Cref{alg:bwg} is the greedy algorithm for this objective. Thus, the $0.5096$ guarantee follows from \Cref{thm:random-order-partition-greedy}. For the smooth guarantee, independently zeroing each buyer with probability $1-q$ is equivalent to independently omitting the corresponding part $P_i$: a zero-valued buyer can select only a bundle of currently zero-priced items and therefore changes neither the price vector nor the free-disposal value. Hence, \Cref{lem:partition-greedy-smooth} applies. In both cases, $f$'s optimum is exactly the best free-disposal value of any allocation, which is at least $\Opt(v_1, \dots, v_n)$.
\end{proof}

\subsection{Divisible-Item Free-Disposal Algorithms}

In this section, we design random-order free-disposal algorithms for divisible items. We present two algorithms under different assumptions on the valuations, where a stronger assumption of the underlying valuations leads to a better final competitive ratio.

\begin{theorem} \label{thm:continuous-bwg}
    For buyers with capped-XOS valuations, \Cref{alg:continuous-bwg} is $0.5096$-competitive and $\Gamma$-smooth-competitive under the free-disposal value when buyers arrive in a uniformly random order, where $\Gamma(q) = q - q^2/2$. Furthermore, if the valuations are also concave, then the algorithm can be implemented in polynomial time using polynomially many supporting-price queries.
\end{theorem}

\begin{theorem} \label{thm:div-free disposal}
    For buyers with DR-submodular and concave valuations, \Cref{alg:div-free disposal} is $(1-1/e) \approx 0.6321$-competitive under the free-disposal value, even under an adversarial buyer arrival order.
\end{theorem}

Before describing the algorithms, we define a quantity that will help us decompose the free-disposal value. For any allocation vectors $\vx_1, \dots, \vx_n$, item $j$, and weight level $w$, let $y_{j}(w)$ be the fraction of item $j$ allocated at supporting prices whose value is at least $w$. Formally,
\[y_j(w) = \min\left\{1, \sum_{i=1}^n x_{i, j} \cdot \bm{1}[q_{i, j}(\vx_i) \ge w]\right\}.\]
Then, the free-disposal value is exactly
\[\FD(\vx_1, \dots, \vx_n) = \sum_{j=1}^m \int_0^\infty y_{j}(w) \, dw.\]
We then let $y_j = \int_0^\infty y_j(w) \, dw$ be the contribution of item $j$ to the free-disposal value. We also observe that the same definition can be extended to partial allocations, i.e., those for which only some of the $\vx_i$ have been determined.

\subsubsection{A Continuous Extension of Buyer-Wise Greedy}

We introduce the continuous extension of the buyer-wise greedy algorithm (\Cref{alg:bwg}) as follows. Observe that in the indivisible-item case, $p_j^{(i)}$ is exactly the contribution of item $j$ to the free-disposal value at step $i$, since the supporting-price property ensures that $p_j^{(i)}$ increases whenever it is updated. In other words, the set $S_i$ that we choose at step $i$ in \Cref{alg:bwg} is exactly the set that increases the free-disposal value the most.

Following this viewpoint, we state the continuous buyer-wise greedy algorithm as \Cref{alg:continuous-bwg}.

\begin{algorithm}[H]
    \caption{Continuous Buyer-Wise Greedy}
    \label{alg:continuous-bwg}
    \begin{algorithmic}[1]
    \Input{Arriving buyers with valuations $\vv$.}
    \Output{Allocation vectors $\vx_1, \dots, \vx_n$ that need not be feasible.}
    \For {buyer $i \in [n]$}
        \State Choose $\vx_i \in [0, 1]^m$ that maximizes the increment in the free-disposal value. \label[line]{lin:need-concave}
        \State Allocate $\vx_i$ to buyer $i$.
    \EndFor
    \end{algorithmic}
\end{algorithm}

The proof idea of \Cref{thm:continuous-bwg} is similar to that of \Cref{thm:bwg-random-order}; that is, it is essentially a reduction to online monotone submodular maximization over a partition matroid. However, we cannot directly follow the approach used in the indivisible-item case, since there are infinitely many possible choices for $\vx_i$ at any point, making the reduced instance infinite and potentially invalidating the use of \Cref{thm:random-order-partition-greedy}. Instead, we construct a finite instance using \emph{only the possible outcomes for $\vx_i$ produced by running \Cref{alg:continuous-bwg}, together with the hindsight-optimal choice of $\vx_i$}. There are at most $n! + 1$ such choices for $\vx_i$, so the instance is finite.

Formally, we first fix $v_1, \dots, v_n$, and let $\vx^*_1, \dots, \vx^*_n$ be the optimal allocations. For every buyer $i$, we record the set
\[O_i = \{\text{$\vx_i$ produced by \Cref{alg:continuous-bwg} over all possible buyer orders}\} \cup \{\vx^*_i\}.\]

We will only consider the configurations $e_{i, \vx_i}$ for $\vx_i \in O_i$, and there are finitely many of them. As in the indivisible case, let
\[
    \mathcal{N}=\{e_{i,\vx_i}:i\in[n],\vx_i \in O_i\}
\]
and
\[
    P_i=\{e_{i,\vx_i}:\vx_i \in O_i\}.
\]

We now define the reward function $f$ as follows. For $\mathcal{I}\subseteq\mathcal{N}$, let
\[
    f(\mathcal{I})
    := \FD(\{\vx_i : e_{i, \vx_i} \in \mathcal{I}\}).
\]

The crucial observation here is that any possible output of \Cref{alg:continuous-bwg} is encoded as a feasible subset $\mathcal{I} \subseteq \mathcal{N}$, with the free-disposal value being exactly $f(\mathcal{I})$. Similarly, since $\vx^*_i$ is included in $O_i$, the optimum allocation is also encoded as a feasible subset. Finally, due to the greedy nature of \Cref{alg:continuous-bwg}, we can still show that the behavior of the greedy algorithm from \Cref{thm:random-order-partition-greedy} on the reduced instance is similar to that of \Cref{alg:continuous-bwg} for the same buyer arrival order.

Again, we first prove that $f$ is submodular.

\begin{lemma} \label{lem:f-submodular}
    The function $f$ is normalized, nonnegative, monotone, and submodular.
\end{lemma}

\begin{proof}
    We prove submodularity here. By definition,
    \[f(\mathcal{I}) = \sum_{j=1}^m \int_0^\infty \underbrace{\min\left\{1, \sum_{e_{i, \vx_i} \in \mathcal{I}} x_{i, j} \cdot \bm{1}[q_{i, j}(\vx_i) \ge w]\right\}}_{y_{j, w}(\mathcal{I})} \, dw.\]

    If we can prove that $y_{j, w}(\mathcal{I})$ is submodular for all $j \in [m], w \ge 0$, then we are done. This is straightforward to show: for $\mathcal{I} \subseteq \mathcal{J}$, we have $y_{j, w}(\mathcal{I}) \le y_{j, w}(\mathcal{J})$, so the marginal contribution of $e_{i, \vx_i}$ is
    \begin{align*}
        y_{j, w}(e_{i, \vx_i} \mid \mathcal{I}) &= \min\left\{x_{i, j} \cdot \bm{1}[q_{i, j}(\vx_i) \ge w], 1 - y_{j, w}(\mathcal{I})\right\}
        \\&\ge \min\left\{x_{i, j} \cdot \bm{1}[q_{i, j}(\vx_i) \ge w], 1 - y_{j, w}(\mathcal{J})\right\} = y_{j, w}(e_{i, \vx_i} \mid \mathcal{J})
    \end{align*}
    which is the diminishing-returns definition of submodularity.
\end{proof}

Combining this lemma with \Cref{thm:random-order-partition-greedy} gives us the desired competitive ratio. Regarding runtime, however, even with the concavity of $v_i$, the increment in the free-disposal value is not necessarily concave, so it is not obvious how to find the greedy choice $\vx_i$.

To this end, observe the following: when we add the vector $\vx_i$ to the partial allocation, the free-disposal value is updated to include a fraction $\vc \le \vx_i$ of $\vq_i = \vq_i(\vx_i)$, while the previous partial allocation is \emph{truncated} to have total weight $\bm{1} - \vc$. In particular, for any $\vc \in [0, 1]^m$, if we define the \emph{truncated free-disposal value}
\[
    \FD_{\vc}(\vx_1, \dots, \vx_{i - 1}) := \sum_{j=1}^m \int_0^\infty \min\left\{c_j, \sum_{i'=1}^{i - 1} x_{i', j} \cdot \bm{1}[q_{i', j}(\vx_i') \ge w]\right\} \, dw
\]
then the observation above shows that
\[
    \FD(\vx_1, \dots, \vx_i) = \max_{\bm{0} \le \vc \le \vx_i} \inner{\vq_i(\vx_i)}{\vc} + \FD_{\bm{1} - \vc}(\vx_1, \dots, \vx_{i - 1}).
\]

We now claim that, for the greedy choice of $\vx_i$, we include the fraction $\vx_i$ of $\vq_i$ in the new free-disposal value while truncating the previous partial allocation to weight $\bm{1} - \vx_i$. This claim follows intuitively from the supporting-price property: if the objective used only a fraction $\vc \neq \vx_i$ of $\vq_i$, then choosing $\vc$ instead of $\vx_i$ would yield at least as much value because $\vc \le \vx_i$. We formalize this claim as follows.

\begin{claim}\label{cl:surrogate-obj}
    Define the function $f : [0, 1]^m \to \R_{\ge 0}$ by
    \[f(\vx) := v_i(\vx) + \FD_{\bm{1} - \vx}(\vx_1, \dots, \vx_{i - 1}).\]

    Then, any maximizer of $f$ is also a valid greedy choice at \Cref{lin:need-concave}.
\end{claim}

\begin{proof}
    First, note that for all $\vx$,
    \begin{align*}
        \FD(\vx_1, \dots, \vx_{i - 1}, \vx) &= \max_{\bm{0} \le \vc \le \vx} \inner{\vq_i(\vx)}{\vc} + \FD_{\bm{1} - \vc}(\vx_1, \dots, \vx_{i - 1}) \\
        &\ge \inner{\vq_i(\vx)}{\vx} + \FD_{\bm{1} - \vx}(\vx_1, \dots, \vx_{i - 1}) \\
        &= v_i(\vx) + \FD_{\bm{1} - \vx}(\vx_1, \dots, \vx_{i - 1}) = f(\vx).
    \end{align*}

    Let $\vx^*_i$ be a maximizer of $f$, and let $\vx'_i$ be a maximizer of $\vx \mapsto \FD(\vx_1, \dots, \vx_{i - 1}, \vx)$, i.e., the vector chosen at \Cref{lin:need-concave}. On the one hand, by the pointwise bound above, we have
    \begin{equation} \label{eq:fd-to-f}
        f(\vx^*_i) \le \FD(\vx_1, \dots, \vx_{i - 1}, \vx^*_i) \le \FD(\vx_1, \dots, \vx_{i - 1}, \vx'_i).
    \end{equation}

    On the other hand, letting $\vc = \argmax_{\bm{0} \le \vc \le \vx'_i} \inner{\vq_i(\vx'_i)}{\vc} + \FD_{\bm{1} - \vc}(\vx_1, \dots, \vx_{i - 1})$, we have
    \begin{equation} \label{eq:f-to-fd}
    \FD(\vx_1, \dots, \vx_{i - 1}, \vx'_i) = \inner{\vq_i(\vx'_i)}{\vc} + \FD_{\bm{1} - \vc}(\vx_1, \dots, \vx_{i - 1}) \le v_i(\vc) + \FD_{\bm{1} - \vc}(\vx_1, \dots, \vx_{i - 1}) \le f(\vx^*_i)
    \end{equation}
    where the inequality is due to the supporting-price property for $\vc \le \vx'_i$. Therefore, equality must hold throughout \Cref{eq:f-to-fd,eq:fd-to-f}; in particular, $\FD(\vx_1, \dots, \vx_{i - 1}, \vx^*_i) = \FD(\vx_1, \dots, \vx_{i - 1}, \vx'_i)$. Therefore, $\vx^*_i$ is also a maximizer of the increment in the free-disposal value.
\end{proof}

We are finally ready to prove \Cref{thm:continuous-bwg}.

\begin{proof}[Proof of \Cref{thm:continuous-bwg}]
    Map any buyer arrival order to the corresponding arrival order of the parts $P_i$. The outputs of \Cref{alg:continuous-bwg} and the greedy algorithm of \Cref{thm:random-order-partition-greedy} are then the same. Indeed, when buyer $i$ arrives, \Cref{alg:continuous-bwg} chooses the vector $\vx_i$ that increases the free-disposal value the most among all options in $[0, 1]^m$. This vector is also the greedy choice among all options in $O_i$ and is therefore chosen by the greedy algorithm. Thus, the expected free-disposal value of \Cref{alg:continuous-bwg} is exactly the expected value of greedy on $f$. Furthermore, since $\vx^*_i \in O_i$ for all $i$, the optimal value of $f$ is at least the hindsight optimum $\Opt(v_1, \dots, v_n)$. The competitive and smooth competitive guarantees of \Cref{alg:continuous-bwg} then follow from \Cref{thm:random-order-partition-greedy}.

    For the smooth guarantee, independently zeroing buyers is equivalent to omitting their corresponding parts by reasoning similar to that in the proof of \Cref{thm:bwg-random-order}. The retained buyers appear in a uniformly random relative order. Moreover, every choice made in such a pruned execution belongs to $O_i$, since its retained-buyer prefix can be realized in a full order by placing the omitted buyers afterward. Therefore, \Cref{lem:partition-greedy-smooth} gives $\Gamma(q)=q-q^2/2$.

    It remains to establish the runtime guarantee under the additional assumption that all the $v_i$ are concave. By \Cref{cl:surrogate-obj}, it is sufficient to find a maximizer of $f$. Furthermore, observe that $\FD_{\vc}(\vx_1, \dots, \vx_{i - 1})$ is concave in $\vc$, so $f(\vx)$ is also concave. Finally, for a fixed $j$, the map
    \[w \mapsto \min\left\{c_j, \sum_{i'=1}^{i - 1} x_{i', j} \cdot \bm{1}[q_{i', j}(\vx_i') \ge w]\right\}\] is a step function with at most $n$ breakpoints. Therefore, assuming value and first-order oracles for $v$, we can also evaluate $f$ and apply standard concave-optimization methods to maximize it to arbitrary additive accuracy in polynomial time.
\end{proof}

\subsubsection{A \texorpdfstring{$(1-1/e)$}{(1-1/e)}-Competitive Algorithm}

Let us consider the following primal LP for the hindsight optimum. Here, the optimum allocates to buyer $i$ a measure $\mu_i$ over the possible allocation vectors $\vx \in [0, 1]^m$.
\begin{empheq}[box=\fbox]{align*}
    \text{maximize} \quad\quad\quad & \sum_{i=1}^n \int_{\vx \in [0, 1]^m} v_i(\vx) \, d \mu_{i, \vx} \\
    \text{s.t.} \quad\quad\quad & \int_{\vx \in [0, 1]^m} \, d\mu_{i, \vx} \le 1 & \forall i \in [n] \\
    & \sum_{i=1}^n \int_{\vx \in [0, 1]^m} x_j \, d \mu_{i, \vx} \le 1 & \forall j \in [m] \\
    & d \mu_{i, \vx} \ge 0
\end{empheq}

Its dual is
\begin{empheq}[box=\fbox]{align*}
    \text{minimize} \quad\quad\quad & \sum_{i=1}^n \beta_i + \sum_{j=1}^m \alpha_j \\
    \text{s.t.} \quad\quad\quad & \beta_i + \inner{\vx}{\valpha} \ge v_i(\vx) & \forall \vx \in [0, 1]^m, i \in [n] \\
    & \beta_i, \alpha_j \ge 0
\end{empheq}

We will use the online primal-dual framework to design and analyze our algorithms. To facilitate the analysis, we consider the following modification to the primal and dual variables.
\begin{itemize}
    \item For each item $j$ and each weight level $w$, we define a split-dual variable $\alpha_{j}(w)$, and we maintain the invariant throughout the algorithm that
    \[\alpha_j = \int_0^\infty \alpha_{j}(w)\, dw.\]
    Therefore, we will only specify the updates to the split-dual variables $\alpha_j(w)$.
\end{itemize}

We provide \Cref{alg:div-free disposal} below, for the function $f(x) = \frac{e^{x - 1} - 1/e}{\gamma}$ and $\gamma = 1 - 1/e$. It is straightforward to verify that $f$ is increasing and continuous, that $f(0) = 0$, $f(1) = 1$, and $f'(x) + 1 - f(x) = 1 / \gamma$ for all $x \in [0, 1]$. We want to prove that the algorithm is $\gamma$-competitive. Following the online primal-dual analysis framework, we shall prove the following two claims.

\begin{itemize}
    \item \textbf{Approximate Reverse Weak Duality.} At every iteration, the primal increment is at least $\gamma$ times the dual increment.
    \item \textbf{Dual Feasibility.} At the end, the dual program is feasible.
\end{itemize}

\suppressfloats[t]
\begin{algorithm}[!htb]
    \caption{Free-Disposal Algorithm for Divisible Items}\label{alg:div-free disposal}
    \begin{algorithmic}[1]
    \Input{Arriving buyers with valuations $\vv$.}
    \Output{Allocation vectors $\vx_1, \dots, \vx_n$ that need not be feasible.}
    \State Initialize $y_{j}^{(0)}(w) \gets 0$ and $\alpha_{j}^{(0)}(w) \gets f(y_{j}^{(0)}(w)) = 0$ for all $j \in [m]$ and $w \ge 0$.
    \For {buyer $i \in [n]$}
        \State Let $\vx_i \in [0, 1]^m$ be the vector that satisfies \Cref{cond:fixed-point}. \label[line]{lin:x}
        \For {item $j \in [m]$ and all $w \in [0, q_{i, j}(\vx_i)]$}
            \State $y_{j}^{(i)}(w) \gets \min\{1, y_{j}^{(i - 1)}(w) + x_{i, j}\}$.
            \State $\alpha_{j}^{(i)}(w) \gets f(y_{j}^{(i)}(w))$.
        \EndFor
        \State $\beta_i \gets v_i(\vx_i) - \inner{\vx_i}{\valpha^{(i)}}$.
        \State Allocate $\vx_i$ to buyer $i$.
    \EndFor
    \end{algorithmic}
    \hrulefill
    \begin{condition}[Fixed-point condition] \label{cond:fixed-point}
        For any $\vx \in [0, 1]^m$, let $p_j$ be the hypothetical resulting value of $\alpha_j^{(i)} = \int_0^\infty \alpha_j^{(i)}(w) \, dw$ had we chosen $\vx$ at \Cref{lin:x}. Then, we require
        \[\vx \in \underbrace{\displaystyle\argmax_{\vx' \in [0, 1]^m} \; v_i(\vx') - \inner{\vx'}{\vp}}_{\varphi(\vx)}.\]
    \end{condition}
\end{algorithm}

Before proving that \Cref{cond:fixed-point} always admits a feasible fixed point, we first prove that the algorithm satisfies both approximate reverse weak duality and dual feasibility, assuming that $\vx_i$ is well-defined.

\begin{lemma} \label{lem:weak-dual-div}
    At every iteration, the primal increment is at least $\gamma$ times the dual increment.
\end{lemma}

\begin{proof}
    At iteration $i$, the primal increase is exactly
    \[\sum_{j=1}^m \int_0^\infty \left(y_j^{(i)}(w) - y_j^{(i - 1)}(w)\right) \, dw = \sum_{j=1}^m \int_0^\infty \int_{y_j^{(i - 1)}(w)}^{y_j^{(i)}(w)} 1 \, dt \, dw.\]

    On the dual side, since we maintain $\alpha_j^{(i)}(w) = f(y_j^{(i)}(w))$ throughout, the increase in $\valpha$ is exactly
    \begin{align*}
    &\sum_{j=1}^m \int_0^\infty \left(\alpha_j^{(i)}(w) - \alpha_j^{(i - 1)}(w)\right) \, dw \\
    &= \sum_{j=1}^m \int_0^\infty \left(f(y_j^{(i)}(w)) - f(y_j^{(i - 1)}(w))\right) \, dw \\
    &= \sum_{j=1}^m \int_0^\infty \int_{y_j^{(i - 1)}(w)}^{y_j^{(i)}(w)} f'(t) \, dt \, dw
    \end{align*}
    while the increase in $\vbeta$ is
    \begin{align*}
    \beta_i &= v_i(\vx_i) - \inner{\vx_i}{\valpha^{(i)}} \\
        &= \sum_{j=1}^m  x_{i, j} \left(q_{i, j}(\vx_i) - \int_0^\infty f(y_j^{(i)}(w)) \, dw \right) \\
        &\le \sum_{j=1}^m \int_0^{q_{i, j}(\vx_i)}  x_{i, j} (1 - f(y_j^{(i)}(w))) \, dw  \\
        &\le \sum_{j=1}^m \int_0^{q_{i, j}(\vx_i)}  \int_{y_j^{(i - 1)}(w)}^{y_j^{(i)}(w)} (1 - f(t)) \, dt \, dw  \\
        &\le \sum_{j=1}^m \int_0^\infty \int_{y_j^{(i - 1)}(w)}^{y_j^{(i)}(w)} (1 - f(t)) \, dt \, dw
    \end{align*}
    where for the second inequality, if $y_j^{(i)}(w) = 1$, then $1 - f(y_j^{(i)}(w)) = 0$; otherwise, if $y_j^{(i)}(w) < 1$, then we have $y_j^{(i)}(w) - y_j^{(i - 1)}(w) = x_{i, j}$, and $1 - f(t)$ is a decreasing function of $t$.

    Finally, since for all $t \in [0, 1]$ we have $1 - f(t) + f'(t) = 1/\gamma$ by construction, the lemma follows.
\end{proof}

\begin{lemma} \label{lem:dual-feas-div}
    At the end, we have $\beta_i + \inner{\vx}{\valpha^{(n)}} \ge v_i(\vx)$ for all $\vx \in [0, 1]^m$ and $i \in [n]$.
\end{lemma}

\begin{proof}
    
    Since $f$ is an increasing function and $y_j^{(i)}(w)$ is non-decreasing in $i$, $\valpha^{(i)}$ is also non-decreasing in $i$. Therefore, it suffices to prove that $\beta_i + \inner{\vx}{\valpha^{(i)}} \ge v_i(\vx)$ for all $\vx \in [0, 1]^m$. This is immediately implied by \Cref{cond:fixed-point}, since
    \[\beta_i = v_i(\vx_i) - \inner{\vx_i}{\valpha^{(i)}} = \max_{\vx' \in [0, 1]^m} v_i(\vx') - \inner{\vx'}{\valpha^{(i)}} \ge v_i(\vx) - \inner{\vx}{\valpha^{(i)}}.\qedhere\]
\end{proof}

We now prove the fixed-point condition, which implies the algorithm is well-defined. The proof uses the assumption that the valuations are DR-submodular (in particular, the supporting-price map is continuous) and concave.

\begin{lemma} \label{lem:fixed-point}
    For any valid state of the algorithm before \Cref{lin:x}, there exists an $\vx \in [0, 1]^m$ that satisfies \Cref{cond:fixed-point}.
\end{lemma}

\begin{proof}
    We prove this via the Kakutani fixed-point theorem applied to $\varphi$. To this end, we first show that the map $\vx \mapsto \vp$ is continuous. By definition of the algorithm, this map is given coordinatewise by
    \[p_j = \alpha_{j}^{(i - 1)} + \int_0^{q_{i, j}(\vx)} f\left(\min\{1, y_{j}^{(i - 1)}(w) + x_j\}\right) - f(y_{j}^{(i - 1)}(w)) \, dw.\]

    Since $v_i$ is DR-submodular, \Cref{prop:dr-submod} implies that $q_{i, j}(\vx)$ is continuous; it is also bounded because $v_i$ is bounded. Therefore, the map $\vx \mapsto \vp$ is continuous, as is the map $(\vx, \vx') \mapsto v_i(\vx') - \inner{\vx'}{\vp}$.

    Next, note that
    \begin{itemize}
        \item The domain $[0, 1]^m$ of $\vx$ is a nonempty, compact, and convex subset of $\R^m$,
        \item $\varphi(\vx)$ is convex and nonempty, as the function $v_i(\vx') - \inner{\vx'}{\vp}$ is finite and concave in $\vx'$ due to the finiteness and concavity of $v_i$, so the maximizers exist and form a convex set, and
        \item $\varphi(\vx)$ has a closed graph due to the continuity of the map $(\vx, \vx') \mapsto v_i(\vx') - \inner{\vx'}{\vp}$.
    \end{itemize}

    Therefore, by the Kakutani fixed-point theorem, $\varphi$ has a fixed point.
\end{proof}

\begin{proof}[Proof of \Cref{thm:div-free disposal}]
    This follows directly from \Cref{lem:weak-dual-div,lem:dual-feas-div,lem:fixed-point}.
\end{proof}

% Single-Sample Single-Item Prophet Inequality

\section{Single-Sample Single-Item Prophet Inequalities}\label{sec:sssi}
In this section, we provide $\SSSI$ algorithms for both the divisible- and indivisible-item cases, each admitting the optimal competitive ratio of $1/2$. We also consider the problem of the \emph{single-sample prophet inequality for a concave divisible item}, which is an important generalization of the single-sample prophet inequality for fractional knapsack.
\subsection{Single-Sample Prophet Inequality for an Indivisible Item}

For the indivisible-item case, the algorithm for the single-sample prophet inequality is a classical result of \cite{rubinstein20samples}. For completeness, we include a modification of their algorithm for our variant. Recall that a buyer is characterized by a pair $(q_i, x_i)$, where $q_i \in \R_{\ge 0}$ is the buyer's value for the item and $x_i \in \{0, 1\}$ indicates whether the buyer arrives.

\suppressfloats[t]
\begin{algorithm}[!htb]
    \caption{Single-Sample Prophet Inequality for an Indivisible Item \cite{rubinstein20samples}}\label{alg:sspi-indiv-item}
    \begin{algorithmic}[1]
    \Input{Sample valuations $\vs$ and arriving buyers with real valuations $\vr$.}
    \Output{Feasible allocations $x_1, \dots, x_n$.}
    \State Let $\tau \gets \max_{x(s_i) = 1} q(s_i)$.
    \State Accept the first real buyer $i$ for whom $x(r_i) = 1$ and $q(r_i) > \tau$.
    \end{algorithmic}
\end{algorithm}

We state the following result without proof; it is a simple modification of the corresponding result proved in \cite{rubinstein20samples}.

\begin{theorem}[\cite{rubinstein20samples}] \label{thm:sspi-indiv-item}
    \Cref{alg:sspi-indiv-item} is a $1/2$-competitive polynomial-time $\SSSI$ for the indivisible-item case. Furthermore, this is the best competitive ratio achievable by any algorithm.
\end{theorem}

\subsection{Single-Sample Prophet Inequality for Fractional Knapsack}

For the divisible-item case, we consider the \emph{single-sample prophet inequality for fractional knapsack}. Recall that each buyer $i$ is represented by a knapsack item $(q_i, x_i)$, where $q_i$ is the \emph{density} of the item and $x_i$ is its weight. Upon the arrival of item $i$, an algorithm must immediately choose a fraction $x_i' \le x_i$ of the item while ensuring that the capacity-$1$ knapsack is not overfilled: $\sum_{i=1}^n x_i' \le 1$. Finally, the algorithm's reward is $\sum_{i=1}^n x_i' q_i$. The benchmark is the offline optimal fractional knapsack, and it is well known that the greedy algorithm in decreasing order of density achieves this benchmark. For ease of notation, for a knapsack item $u$, let $q(u)$ and $x(u)$ be the density and weight of the item, respectively.

We prove that, in the single-sample setting, \Cref{alg:sspi-frac-knapsack} achieves a competitive ratio of $1/2$ with respect to the hindsight optimum. This result has already been proved in \cite{nuti2026sample}, but for the model in which $s_i$ and $r_i$ are drawn i.i.d.\ from a distribution $\D_i$. Their proof does not apply in the Game of Googol setting, so we prove the theorem from scratch.

\suppressfloats[t]
\begin{algorithm}[!htb]
    \caption{Single-Sample Prophet Inequality for Fractional Knapsack \cite{nuti2026sample}}\label{alg:sspi-frac-knapsack}
    \begin{algorithmic}[1]
    \Input{Sample knapsack items $\vs$ and arriving knapsack items $\vr$.}
    \Output{Feasible knapsack allocations $x'_1, \dots, x'_n$.}
    \State Compute the optimal fractional knapsack of $\vs$; let the resulting knapsack be $K_s$.
    \State Let $\tau_s$ be the density of a uniformly random quantile in $K_s$.
    \State Let $R \gets 1$.
    \For{item $i \in [n]$}
        \If {$q(r_i) > \tau_s$}
            \State Accept the maximum feasible fraction of $r_i$, i.e., $x_i' \gets \min\{R, x(r_i)\}$.
            \State $R \gets R - x_i'$.
        \EndIf
    \EndFor
    \State Let $K_r$ be the optimal fractional knapsack of $\vr$, and $\tau_r$ be the density of a uniformly random quantile in $K_r$. \Comment{Only for analysis.}
    \end{algorithmic}
\end{algorithm}

\begin{theorem}\label{thm:sspi-frac-knapsack}
    \Cref{alg:sspi-frac-knapsack} is a $1/2$-competitive polynomial-time $\SSSI$ for the divisible-item case. Furthermore, this is the best competitive ratio achievable by any algorithm.
\end{theorem}

We will not prove the optimality statement, as it follows from a standard argument for the single-item prophet inequality. For the rest of this section, we assume that $\sum_{i=1}^n x(s_i) \ge 1$ and $\sum_{i=1}^n x(r_i) \ge 1$ always hold. This assumption is without loss of generality, as one can insert a dummy item of weight $1$ and density $0$ at the end of both the sample and real profiles without affecting the algorithm. We further assume that all items in $\vs$ and $\vr$ have different densities. Otherwise, the algorithm and analysis still apply under a suitable tie-breaking rule, such as an infinitesimal random perturbation of the densities.

For any $t \ge 0$, let $K_r(t)$ be the fraction of the prophet knapsack for $\vr$ whose density is larger than $t$, and let $A(t)$ be the fraction of the algorithm's output whose density is larger than $t$. We define $K_s(t)$ analogously for the optimal fractional knapsack of $\vs$. If we can prove that $\E[A(t)] \ge \frac{1}{2} \cdot \E[K_r(t)]$ for almost all $t \ge 0$, then we are done, as

\[\Alg = \int_0^\infty \E[A(t)] \, dt \ge \frac{1}{2} \int_0^\infty \E[K_r(t)] \, dt = \frac{1}{2} \cdot \Opt.\]

Fix any $t \ge 0$. By the definition of $\tau_r$ and the greedy algorithm that computes the optimal fractional knapsack, we have the following observation.

\begin{observation}\label{obs:knapsack-thres}
    For any realization of $\vr$, we have $K_r(t) = \Pr[\tau_r > t \mid \vr]$. Furthermore, this is exactly $\min\{H, 1\}$, where $H := \sum_{i=1}^n x(r_i) \bm{1}[q(r_i) > t]$ is the total weight of the items in $\vr$ whose density is above $t$.
\end{observation}

By symmetry, we now know that $\E[K_r(t)] = \E[\tau_s > t]$. We first decompose $\E[A(t)]$ into two cases.

\paragraph{Case $\tau_s \ge t$.} Here, every portion of the knapsack accepted by the algorithm has density above $\tau_s \ge t$, so

\[\E[A(t) \cdot \bm{1}[\tau_s \ge t]] = \E[K_r(\tau_s) \cdot \bm{1}[\tau_s \ge t]].\]

\paragraph{Case $\tau_s < t$.} The danger here is that the algorithm may accept too many fractions with medium density, thereby blocking items with high density. Therefore, if we define
\[M := \sum_{i=1}^n x(r_i) \cdot \bm{1}[\tau_s < q(r_i) \le t], \quad\quad\quad H := \sum_{i=1}^n x(r_i) \cdot \bm{1}[q(r_i) > t]\]
then the fraction of items with density above $t$ accepted in this case is at least
\begin{align*}
    \min\{H, (1 - M)^+\} &= (1 - M)^+ - (1 - M - H)^+ \\
    &= (1 - M)^+ - (1 - \min\{M + H, 1\})\\
    &= (1 - M)^+ - (1 - K_r(\tau_s))
\end{align*}
where the third equality follows from \Cref{obs:knapsack-thres}.

Combining the two cases, we have
\begin{align}
    \E[A(t)] \ge &\; \E[K_r(\tau_s) \cdot \bm{1}[\tau_s \ge t]] \label{eq:first}\\
        & - \E[(1 - K_r(\tau_s)) \cdot \bm{1}[\tau_s < t]] \label{eq:second} \\
        & + \E[(1 - M)^+ \cdot \bm{1}[\tau_s < t]]. \label{eq:annoying}
\end{align}

We will now bound each term individually.

\begin{lemma}\label{lem:eq-first}
Term (\ref{eq:first}) satisfies
    \[2 \E[K_r(\tau_s) \cdot \bm{1}[\tau_s \ge t]] = \Pr[\tau_s \ge t, \tau_r \ge t].\]
\end{lemma}

\begin{proof}
    By symmetry,
    \begin{align*}
        2 \E[K_r(\tau_s) \cdot \bm{1}[\tau_s \ge t]] &= \E[K_r(\tau_s) \cdot \bm{1}[\tau_s \ge t]] + \E[K_s(\tau_r) \cdot \bm{1}[\tau_r \ge t]] \\
        &= \Pr[\tau_r > \tau_s, \tau_s \ge t] + \Pr[\tau_s > \tau_r, \tau_r \ge t] \\
        &= \Pr[\tau_s \ge t, \tau_r \ge t]
    \end{align*}
    where the last equality uses the fact that $\tau_s = \tau_r$ can never occur because $\vs$ and $\vr$ have different densities.
\end{proof}

The following lemma is proved similarly.

\begin{lemma}\label{lem:eq-second}
Term (\ref{eq:second}) satisfies
    \[2 \E[(1 - K_r(\tau_s)) \cdot \bm{1}[\tau_s < t]] = \Pr[\tau_s < t, \tau_r < t].\]
\end{lemma}

Finally, we have the following.

\begin{lemma}\label{lem:annoying-eq}
If $t$ does not coincide with any density, then term (\ref{eq:annoying}) satisfies
\[2 \E[(1 - M)^+ \cdot \bm{1}[\tau_s < t]] \ge \Pr[\tau_s < t],\]
where $M = \sum_{i=1}^n x(r_i) \cdot \bm{1}[\tau_s < q(r_i) \le t]$.
\end{lemma}
\begin{proof}
Let the \emph{truncated instances} $\ovl{\vs}$ and $\ovl{\vr}$ be obtained by removing all items with density above $t$ from $\vs$ and $\vr$, respectively. Let $\ovl{K_s},\ovl{K_r}$ be the optimal knapsacks of the truncated instances, and let $\ovl{\tau_s}$ and $\ovl{\tau_r}$ be the random quantiles from these knapsacks. Whenever several random quantiles are used below, draw them independently conditional on the sample-real orientations.

Let $Z$ denote the vector of independent orientation coins that determine $\vs$ and $\vr$.
Conditional on $\tau_s<t$ and $Z$, the quantile $\tau_s$ is drawn from a top-density prefix of $\ovl{K_s}$ and hence stochastically dominates $\ovl{\tau_s}$. Since $\ovl{\tau_r}$ is drawn independently of these two quantiles conditional on $Z$, for every $Z$,
\begin{align*}
\Pr[\ovl{\tau_r}<\tau_s,\ \tau_s<t\mid Z]
&=\Pr[\tau_s<t\mid Z] \Pr[\ovl{\tau_r}<\tau_s\mid \tau_s<t,Z]\\
&\ge \underbrace{\Pr[\tau_s<t\mid Z]}_{f(Z)} \underbrace{\Pr[\ovl{\tau_r} < \ovl{\tau_s} \mid Z]}_{g(Z)}.
\end{align*}

We next apply the FKG inequality only to buyers whose two densities straddle $t$. Let $S$ be their orientation coins, with $1$ meaning that the lower-density item is the sample. Let $L$ be the coins of the ``low'' buyers whose two densities are below $t$, and let $H$ be the coins of the ``high'' buyers whose two densities are above $t$. Then $f(Z)=f(S,H)$ and $g(Z)=g(S,L)$. Indeed, because the sample knapsack is filled to capacity after adding the dummy item, $\Pr[\tau_s<t\mid Z]$ is exactly the residual capacity left after packing the sample items of density above $t$. The orientation of a low buyer can change which below-$t$ item is in the sample, but cannot change this residual capacity, so $f$ is independent of $L$. On the other hand, both possible items of a high buyer are removed from the sample and real instances by truncation, regardless of their orientation. Thus, $g$ is independent of $H$.

Define $F(S):=\E_H[f(S,H)]$ and $G(S):=\E_L[g(S,L)]$. Since $L$ and $H$ are independent conditional on $S$, we have $\E[f(Z)g(Z)\mid S]=F(S)G(S)$, so $\E[f(Z)g(Z)]=\E[F(S)G(S)]$.

Both $F$ and $G$ are non-decreasing. Indeed, changing a coordinate of $S$ from $0$ to $1$ removes an above-$t$ item from the sample, so $f$ increases. After truncation, the same change adds the other item to the sample and removes it from the real instance, so $\ovl{\tau_s}$ increases and $\ovl{\tau_r}$ decreases in stochastic order; hence, $g$ also increases. The FKG inequality therefore gives
\begin{align*}
    \E[f(Z)g(Z)] &= \E[F(S)G(S)] \\
    &\ge \E[F(S)]\E[G(S)] =\Pr[\tau_s<t]\Pr[\ovl{\tau_r}<\ovl{\tau_s}] \\
    &= \frac{1}{2}\Pr[\tau_s<t],
\end{align*}
where the last equality follows from symmetry and the absence of tied densities in $\ovl{\vs}$ and $\ovl{\vr}$, which imply that $\Pr[\ovl{\tau_r}<\ovl{\tau_s}]=1/2$.

Finally, $\min\{M,1\}=\ovl{K_r}(\tau_s)$, and thus 
\[
    \E[(1-M)^+\mathbf 1[\tau_s<t]] =\Pr[\ovl{\tau_r}<\tau_s,\ \tau_s<t] \ge \E[f(Z)g(Z)] \ge \frac{1}{2} \Pr[\tau_s<t]. \qedhere
\]
\end{proof}

We are ready to prove the final theorem.

\begin{proof}[Proof of \Cref{thm:sspi-frac-knapsack}] 
    
Combining \Cref{lem:eq-first,lem:eq-second,lem:annoying-eq}, we have
\begin{align*} 
    2\E[A(t)] &\ge 2\left(\E[K_r(\tau_s) \cdot \bm{1}[\tau_s \ge t]] - \E[(1 - K_r(\tau_s)) \cdot \bm{1}[\tau_s < t]] + \E[(1 - M)^+ \cdot \bm{1}[\tau_s < t]]\right) \\
        &\ge \Pr[\tau_s \ge t, \tau_r \ge t] - \Pr[\tau_s < t, \tau_r < t] + \Pr[\tau_s < t] \\ 
        &= \Pr[\tau_r \ge t] \ge \Pr[\tau_r > t] = \E[K_r(t)],
    \end{align*}
for almost all $t \ge 0$, and therefore the theorem follows by integrating the bound above.
\end{proof}

% \begin{observation}
%     For any realization of $\vr$ (and by extension of $\vs$), $\tau_s$, $\ovl{\tau_s}$, $\tau_r$, $\ovl{\tau_r}$ are conditionally mutually independent. Furthermore, $\tau_s$ stochastically dominates $\ovl{\tau_s}$ and $\tau_r$ stochastically dominates $\ovl{\tau_r}$.
% \end{observation}

% \begin{proof}
%      The second item follows from the fact that $\ovl{K_s}$ is the optimal knapsack of the truncated instance $\ovl{\vs} \subseteq \vs$ while $K_s$ is the optimal knapsack of the full instance $\vs$, and similarly to $\ovl{K_r}$ and $K_r$.
% \end{proof}

\subsection{Single-Sample Prophet Inequality for a Concave Divisible Item}

Finally, we consider the single-sample prophet inequality for a concave divisible item. The setting is similar to the fractional-knapsack setting, with the generalization that every valuation function $v : [0, 1] \to \R_{\ge 0}$ is concave. Indeed, a fractional-knapsack item $(q, x)$ can be represented by the concave valuation function $f(x') := q \min\{x, x'\}$.

A natural starting point would be to generalize \Cref{alg:sspi-frac-knapsack} as follows: from the samples $\vs$, build a function $f_s : [0, 1] \to \R_{\ge 0}$ where
\[
f_s(t) := \max_{\substack{x_1, \dots, x_n \ge 0 \\ \sum_{i=1}^n x_i = t}} \sum_{i=1}^n s_i(x_i).
\]
This function is concave in $t$. Then, take $\tau_s = f'_s(u)$, where $u$ is drawn uniformly from $[0, 1]$. Finally, for every arriving real valuation $r_i$, accept the maximum feasible $x_i$ such that $r_i'(x_i) > \tau_s$. However, we show that this approach is insufficient to obtain a competitive ratio of $1/2$ for this problem.

\begin{lemma} \label{lem:lb-concave-item}
    The natural generalization of \Cref{alg:sspi-frac-knapsack} has a competitive ratio of at most $7/16 = 0.4375 < 1/2$ for concave divisible items.
\end{lemma}

We defer the proof of this lemma to \Cref{app:sssi}. We next provide a simple algorithm that achieves a competitive ratio of $\frac{1-1/e}{2} \approx 0.3161$. This leaves a gap between the achieved ratio and the information-theoretic upper bound of $1/2$. Intuitively, for some large $M > 0$, the algorithm breaks each valuation function into marginal gains of the form $v(k/M) - v((k - 1)/M)$ and then runs $M$ instances of the single-sample indivisible-item prophet inequality, one for each ``slot'' of the marginal gains. However, since the marginal gains of each valuation function are correlated and may affect the aggregate performance of the $M$ instances, the algorithm decouples this correlation by randomly permuting the marginal gains among the $M$ slots. This idea is formally described in \Cref{alg:sspi-concave}.

\suppressfloats[t]
\begin{algorithm}[!htb]
    \caption{Single-Sample Prophet Inequality for a Concave Divisible Item}\label{alg:sspi-concave}
    \begin{algorithmic}[1]
    \Input{Sample valuations $\vs$, arriving buyers with valuations $\vr$, and a large integer parameter $M \in \mathbb{Z}_{> 0}$.}
    \Output{Feasible allocations $x_1, \dots, x_n$.}
    \State Initialize $M$ instances of \Cref{alg:sspi-indiv-item} as $\SSSI_1, \dots, \SSSI_M$.
    \For{buyer $i \in [n]$}
        \State Let $\bm{S}^{(i)}$ be a length-$M$ vector with $S^{(i)}_k \gets s_i(k/M) - s_i((k - 1)/M)$ for $k \in [M]$.
        \State Independently and uniformly randomly permute $\bm{S}^{(i)}$.
        \State Send $(S^{(i)}_k, 1)$ to $\SSSI_k$ for all $k \in [M]$ as the sample for buyer $i$.
    \EndFor
    \For{buyer $i \in [n]$}
        \State Let $\bm{R}^{(i)}$ be a length-$M$ vector with $R^{(i)}_k \gets r_i(k/M) - r_i((k - 1)/M)$ for $k \in [M]$.
        \State Independently and uniformly randomly permute $\bm{R}^{(i)}$.
        \State Send $(R^{(i)}_k, 1)$ to $\SSSI_k$ for all $k \in [M]$ as the real value of buyer $i$, and receive allocation $x^{(i)}_k \in \{0, 1\}$.
        \State Allocate $x_i \gets \frac{1}{M} \sum_{k=1}^M x^{(i)}_k$ to buyer $i$.
    \EndFor
    \end{algorithmic}
\end{algorithm}

\begin{theorem}\label{thm:sspi-concave}
    As $M \to \infty$, the competitive ratio of \Cref{alg:sspi-concave} approaches $\frac{1-1/e}{2} \approx 0.3161$ for the concave divisible-item case.
\end{theorem}

\begin{proof}
Fix $\vr$ and write $\Delta_{i,k}=r_i(k/M)-r_i((k-1)/M)$ for the $k$-th marginal piece of buyer $i$. By concavity, $\Delta_{i,1}\ge\cdots\ge\Delta_{i,M}\ge0$. If buyer $i$ wins $m_i := \sum_{k=1}^M x_k^{(i)}$ of the $M$ $\SSSI$ instances, then the values won by these instances are some $m_i$ entries of this sequence, whereas the algorithm obtains $r_i(m_i/M)=\sum_{j=1}^{m_i}\Delta_{i,j}$. Thus, the algorithm's actual reward is at least the sum of the rewards obtained by the $M$ $\SSSI$ instances. Moreover, each instance allocates at most one piece, so $\sum_{i=1}^n m_i\le M$, and the resulting allocation is feasible.

For any fixed instance $k$, the sample and real marginals presented by a buyer are exchangeable: each is a uniformly random marginal from the corresponding valuation, and the fair coin that swaps the two underlying valuations makes the joint distribution invariant under swapping the sample and real marginals. Hence, this is a mixture of valid Game of Googol instances, so \Cref{alg:sspi-indiv-item} obtains at least half of $\E\left[\max_i R^{(i)}_k\right]$ by \Cref{thm:sspi-indiv-item}. It remains only to compare the sum of these per-instance maxima with the divisible optimum.

Let $N_i(w)=|\{j:\Delta_{i,j}\ge w\}|$, and let $\Opt_M(\vr)$ be the optimum restricted to allocations in multiples of $1/M$. Equivalently, $\Opt_M(\vr)$ is the sum of the largest $M$ numbers among all the $\Delta_{i,j}$: since the marginals of each buyer are decreasing, taking a marginal makes all earlier marginals at least as attractive. In a fixed instance, buyer $i$ contributes a marginal of value at least $w$ with probability $N_i(w)/M$, independently across buyers; intuitively, the only loss is that several valuable marginals may land in the same instance. Therefore, using $1-z\le e^{-z}$ and $1-e^{-z}\ge(1-1/e)\min\{1,z\}$,
\begin{align*}
\E\left[\sum_{k=1}^M\max_i R^{(i)}_k\,\middle|\,\vr\right]
&=M\int_0^\infty\left(1-\prod_{i=1}^n \left(1-\frac{N_i(w)}M\right)\right)dw \\
&\ge(1-1/e)\int_0^\infty\min\left\{M,\sum_{i=1}^n N_i(w)\right\}dw \\
&=(1-1/e)\Opt_M(\vr).
\end{align*}
Combining the two losses and then averaging over $\vr$ gives $\Alg_M\ge\frac{1-1/e}{2} \cdot \E[\Opt_M(\vr)]$. Finally, as $M\to\infty$, $\Opt_M(\vr) \to \Opt(\vr)$, proving the asymptotic ratio.
\end{proof}

For completeness, the following lemma makes the convergence $\Opt_M(\vr) \to \Opt(\vr)$ quantitative. Its proof is deferred to \Cref{app:sssi}.

\begin{lemma}\label{lem:grid-opt-approx}
    For any fixed integer $M \ge 1$, we have
    \[\Opt_M(\vr) \ge \frac{M}{M + n} \cdot \Opt(\vr).\]

    Therefore, for fixed $n$, as $M \to \infty$, we have $\Opt_M(\vr) \to \Opt(\vr)$.
\end{lemma}

% Application
\section{Applications}
\label{sec:final-results}

Finally, we apply our main theorem with specific initializations of the random-order free-disposal algorithm and the single-sample single-item prophet inequality for each setting. In particular, we obtain the following guarantees for the indivisible- and divisible-item settings.

\begin{theorem}\label{thm:sscpi-xos}
    There is a $\frac{1}{6\sqrt{3}} \approx \frac{1}{10.40}$-competitive $\SSCPI$ algorithm for the indivisible-item case under XOS valuations. Furthermore, this algorithm runs in polynomial time and uses a polynomial number of demand and supporting-price queries.
\end{theorem}

\begin{proof}
    Apply \Cref{alg:sscpi} with \Cref{alg:bwg} as the $\FDAlg$ component, \Cref{alg:sspi-indiv-item} as the $\SSSI$ component, and $q = 1 - \frac{1}{\sqrt{3}}$. Then, by \Cref{thm:sscpi,thm:bwg-random-order,thm:sspi-indiv-item}, we obtain the competitive guarantee. Furthermore, since \Cref{alg:bwg} runs in polynomial time and uses polynomially many demand and supporting-price queries while \Cref{alg:sspi-indiv-item} runs in polynomial time, we obtain the runtime guarantee.
\end{proof}

\begin{theorem}\label{thm:sscpi-capped-xos}
    There is a $\frac{1}{6\sqrt{3}} \approx \frac{1}{10.40}$-competitive $\SSCPI$ algorithm for the divisible-item case under capped-XOS valuations. Furthermore, if the valuation functions are concave, then the algorithm runs in polynomial time.
\end{theorem}

\begin{proof}
    Apply \Cref{alg:sscpi} with \Cref{alg:continuous-bwg} as the $\FDAlg$ component, \Cref{alg:sspi-frac-knapsack} as the $\SSSI$ component, and $q = 1 - \frac{1}{\sqrt{3}}$. Then, by \Cref{thm:sscpi,thm:continuous-bwg,thm:sspi-frac-knapsack}, we obtain the competitive guarantee. Furthermore, as \Cref{alg:sspi-frac-knapsack} runs in polynomial time unconditionally, while \Cref{alg:continuous-bwg} runs in polynomial time and uses polynomially many supporting-price queries when the valuation functions are concave (\Cref{thm:continuous-bwg}), we obtain the runtime guarantee.
\end{proof}

% Conclusion
\section{Conclusion}

We developed a general reduction from single-sample combinatorial prophet inequalities to two simpler ingredients: a random-order online algorithm for the corresponding free-disposal problem and a single-sample prophet inequality for one resource. This decomposition separates the combinatorial task of producing valuable provisional allocations from the stochastic task of resolving competition for each resource, and it yields improved guarantees for indivisible XOS valuations and new guarantees for divisible capped-XOS valuations in the stronger Game of Googol model.

The same technique is not limited to the indivisible and fully divisible settings studied here. For example, taking $A=\{0,\frac{1}{k},\frac{2}{k},\dots,1\}$ models a $k$-selection or capacity-$k$ setting, and the reduction can be instantiated whenever suitable free-disposal and single-sample single-item algorithms are available; in fact, \cite{nuti2026static} provided single-sample $k$-selection algorithms for the case where the sample and real values are drawn i.i.d.\ from unknown independent distributions. More broadly, the framework offers a modular route to single-sample prophet inequalities in other allocation settings that admit supporting prices and a per-resource decomposition, leaving the design of sharper primitives in these settings as an interesting direction for future work.

\newpage
\appendix
\section{Extension to Subadditive Valuations} \label{sec:subadd}

In this section, we provide a straightforward extension to subadditive valuations in the indivisible-item setting. Recall that a valuation $v:2^{[m]}\to\R_{\ge 0}$ is subadditive if
\[
    v(S\cup T)\le v(S)+v(T)
    \qquad\text{for all }S,T\subseteq[m].
\]
Such a valuation need not have the exact supporting prices required by \Cref{def:supporting-prices}. We overcome this obstruction by replacing $v$ with its XOS approximation, obtained via the following classical result.

\begin{theorem}[\cite{dobzinski2007two,bhawalkar2011welfare}] \label{thm:subadd-approx-xos}
    Let $v:2^{[m]}\to\R_{\ge 0}$ be normalized, monotone, and subadditive. There is a normalized, monotone XOS valuation $\widehat v$ such that, for every $S\subseteq[m]$,
    \[
        \frac{v(S)}{H_m}\le \widehat{v}(S)\le v(S),
    \]
    where $H_m = \sum_{i=1}^m 1/i$ is the $m$-th harmonic number.
\end{theorem}

Given this, it is straightforward to show that any $\alpha$-competitive prophet inequality for XOS buyers can be turned into an $(\alpha/H_m)$-competitive prophet inequality for subadditive buyers.

\begin{theorem} \label{thm:subadd-sspi-reduction}
    If there is an $\alpha$-competitive single-sample combinatorial prophet inequality for XOS buyers over $m$ indivisible items, then there is an $\alpha/H_m = \Omega(\alpha/\log m)$-competitive single-sample combinatorial prophet inequality for subadditive buyers over the same items.
\end{theorem}

\begin{proof}
    For a subadditive valuation $v$, fix its approximating XOS valuation $\widehat{v}$ from \Cref{thm:subadd-approx-xos}. Let $\Alg_{\mathrm{XOS}}$ be an $\alpha$-competitive single-sample combinatorial prophet inequality for XOS buyers. Given the subadditive samples $\vs$, run $\Alg_{\mathrm{XOS}}$ on the transformed samples $\widehat{\vs}:=(\widehat{s}_1, \dots, \widehat{s}_n)$. When buyer $i$ arrives with real valuation $r_i$, present $\widehat{r}_i$ to $\Alg_{\mathrm{XOS}}$ and allocate to buyer $i$ the bundle $S_i$ selected by $\Alg_{\mathrm{XOS}}$.

    The resulting allocation is feasible because it is exactly the allocation selected by $\Alg_{\mathrm{XOS}}$. Moreover, since $\widehat{r}_i(S)\le r_i(S)$ for every bundle $S$, its welfare under the original valuations satisfies the following pointwise inequality:
    \[
        \sum_{i=1}^n r_i(S_i)
        \ge \sum_{i=1}^n \widehat{r}_i(S_i).
    \]

    Fix a realized real profile $\vr$, and let $(S_1^*,\ldots,S_n^*)$ be an optimal allocation for $\vr$. The lower bound in \Cref{thm:subadd-approx-xos} gives
    \[
        \Opt(\widehat{\vr})
        \ge \sum_{i=1}^n \widehat{r}_i(S_i^*)
        \ge \frac{1}{H_m}\sum_{i=1}^n r_i(S_i^*)
        =\frac{1}{H_m}\Opt(\vr).
    \]
    Therefore, taking expectations over the sample-real assignment and all internal randomness, and then using the $\alpha$-competitiveness of $\Alg_{\mathrm{XOS}}$, we obtain
    \[
        \E\left[\sum_{i=1}^n r_i(S_i)\right]
        \ge \E\left[\sum_{i=1}^n \widehat{r}_i(S_i)\right] 
        \ge \alpha\E\left[\Opt(\widehat{\vr})\right]
        \ge \frac{\alpha}{H_m}\E\left[\Opt(\vr)\right].\qedhere
    \]
\end{proof}

\section{Deferred Proofs from \texorpdfstring{\Cref{sec:prelims}}{Section~\ref*{sec:prelims}}} \label{app:prelims}

\begin{proof}[Proof of \Cref{prop:capped-XOS-sup-price}]
First, suppose that $v$ has a supporting-price map. For each $\vx \in [0, 1]^m$, define the following capped-additive clause for every $\vx' \in [0, 1]^m$:
\[
h_{\vx}(\vx')=\sum_{i=1}^m q_i(\vx) \min\{x_i,x'_i\}.
\]
Consider the family $\{h_{\vx}\}_{\vx \in [0, 1]^m}$. For any $\vx, \vx' \in [0, 1]^m$, write $\vx \wedge \vx'$ for the coordinatewise minimum of $\vx$ and $\vx'$. We then have
\[
h_{\vx'}(\vx)=\inner{\vq(\vx')}{\vx' \wedge \vx} \le v(\vx' \wedge \vx) \le v(\vx)
\]
where the first inequality uses the supporting-price property at $\vx'$, and the second uses monotonicity. On the other hand, $h_{\vx}(\vx) = \inner{\vq(\vx)}{\vx} = v(\vx)$. Therefore, $v(\vx) = \max_{\vx' \in [0, 1]^m} h_{\vx'}(\vx)$, so $v$ is a capped-XOS valuation function with the capped-additive family $\{h_{\vx'}\}_{\vx' \in [0, 1]^m}$.

Conversely, suppose that $v$ is capped-XOS, and fix $\vx\in[0,1]^m$. Choose a capped-additive clause $h(\vy)=\sum_{i=1}^m a_i\min\{y_i,c_i\}$ such that $h(\vx)=v(\vx)$. Define
\[
q_i =
a_i \cdot \frac{\min\{x_i,c_i\}}{x_i}
\]
for all $i \in [m]$, with the convention that $q_i = 0$ if $x_i = 0$. It is immediate that $\inner{\vq}{\vx} = \sum a_i \min\{x_i, c_i\} = h(\vx) = v(\vx)$. For any $\bm{0} \le \vx' \le \vx$, we then have
\[\inner{\vq}{\vx'} = \sum_{i=1}^m a_i \cdot \frac{\min\{x_i,c_i\}}{x_i} \cdot x'_i \le \sum_{i=1}^m a_i \cdot \frac{\min\{x_i',c_i\}}{x_i'} \cdot x'_i = \sum_{i=1}^m a_i \cdot \min\{x_i',c_i\} = h(\vx') \le v(\vx')\]
where the first inequality holds because, for a fixed $c_i$, the function $\min\{t, c_i\} / t$ is non-increasing in $t$. Therefore, $\vq$ is a supporting-price vector at $\vx$.
\end{proof}

\section{Deferred Proofs from \texorpdfstring{\Cref{sec:general-reduction}}{Section~\ref*{sec:general-reduction}}} \label{app:general-reduction}

\begin{proof}[Proof of \Cref{lem:combi-pi-to-pi}]
    We first prove the result when $\PAlg$ is deterministic. This is without loss of generality, as we can subsequently take an expectation over the internal randomness of $\PAlg$.

    First, fix the coordination sample vector $\vsc$. Then, for each buyer $i$, $(s^1_i, \dots, s^k_i, r_i)$ forms a random permutation of $k + 1$ fixed valuation functions. This implies that the collection of pairs $(\vq_i^t, \vx_i^t)$ for all samples $t \in [k]$, together with $(\widehat{\vq}_i, \widehat{\vx}_i)$, forms a random permutation of $k + 1$ fixed pairs, since the pairs are generated deterministically from $\PAlg_i(\vsc, s^t_i)$ and $\PAlg_i(\vsc, r_i)$ by the local property of $\PAlg$. Then, by the guarantee of $\kSSI$, for every item $j \in [m]$, we have
    \[
        \E\left[\sum_{i=1}^n x'_{i, j} \widehat{q}_{i, j} \mid \vsc \right] \ge \beta \cdot \E_{\vs^1, \dots, \vs^k, \vr}\left[\Opt(\{(\widehat{q}_{i, j}, \widehat{x}_{i, j})\}_{i \in [n]}) \mid \vsc \right].
    \]
    Summing this over all $j \in [m]$, we have
    \begin{align*}
        \E\left[\sum_{i=1}^n \inner{\vx'_i}{\widehat{\vq}_i} \mid \vsc \right] &\ge \beta \cdot \E_{\vs^1, \dots, \vs^k, \vr}\left[\sum_{j=1}^m \Opt(\{(\widehat{q}_{i, j}, \widehat{x}_{i, j})\}_{i \in [n]}) \mid \vsc \right] \\
        &= \beta \cdot \E_{\vs^1, \dots, \vs^k, \vr}\left[\FD(\widehat{\vx}_1, \dots, \widehat{\vx}_n) \mid \vsc \right].
    \end{align*}

    Now, since $\vx'_i \le \widehat{\vx}_i$ while $\widehat{\vq}_i = \vq_{r_i}(\widehat{\vx}_i)$, the supporting-price property implies that the LHS is at most $\E[\sum_{i=1}^n r_i(\vx'_i) \mid \vsc]$, which is the expected reward of the algorithm given $\vsc$. Finally, taking an expectation over $\vsc$, we have
    \[
        \E\left[\sum_{i=1}^n r_i(\vx'_i)\right] \ge \beta \cdot \E\left[\FD(\widehat{\vx}_1, \dots, \widehat{\vx}_n)\right]
        \ge \alpha \beta \cdot \E[\Opt(\vr)]
    \]
    where we have used the fact that $\PAlg$ is $\alpha$-competitive under the free-disposal value.
\end{proof}

\begin{proof}[Proof of \Cref{lem:biased-mirror}]
Fix the random order $\pi$ and the internal randomness of $\FDAlg$. By relabeling the buyers, assume that $\pi$ is the identity permutation. As in the proof of \Cref{lem:fdalg-to-ss}, for a partial allocation $\vx_1,\dots,\vx_k$ and every item $j\in[m]$ and value $w>0$, write
\[
    \FD(\vx_1,\dots,\vx_k)=\sum_{j=1}^m\int_0^\infty \underbrace{\min\left\{1,\sum_{i=1}^k x_{i,j}\cdot\bm{1}[q_{i,j}(\vx_i)\ge w]\right\}}_{y_{j,w}(\vx_{1..k})}\,dw.
\]
It suffices to prove that
\[
    \E[y_{j,w}(\vx_{1..n})]\ge(1-q)\E[y_{j,w}(\vz_{1..n})]
\]
for every $j$ and $w$. We prove the stronger inequality
\[
    \E[y_{j,w}(\vx_{1..n})]\ge(1-q)\E[y_{j,w}(\vx_{1..n}\cup\vz_{1..n})],
\]
where $y_{j,w}(\vx_{1..n}\cup\vz_{1..n})$ denotes the merged free-disposal contribution, as in the proof of \Cref{lem:fdalg-to-ss}.

Fix a partial realization of the role coins $H_1,\dots,H_{k-1}$ and the samples $s'_1,\dots,s'_{k-1}$ for some $k\in[n]$. This fixes the state $\Sigma^{(k-1)}$ of $\FDAlg$ before buyer $k$, as well as all preceding allocations. Let $f(a):=\min\{1,a\}$, and set
\begin{align*}
    A&=\sum_{i=1}^{k-1}x_{i,j}\cdot\bm{1}[q_{i,j}(\vx_i)\ge w],\\
    B&=\sum_{i=1}^{k-1}\left(x_{i,j}\cdot\bm{1}[q_{i,j}(\vx_i)\ge w]+z_{i,j}\cdot\bm{1}[q_{i,j}(\vz_i)\ge w]\right),\\
    c&=u_{k,j}\cdot\bm{1}[q_{k,j}(\vu_k)\ge w] \quad\text{where $\vu_k=\FDAlg(\Sigma^{(k-1)},v_{k,1})$},\\
    d&=u'_{k,j}\cdot\bm{1}[q_{k,j}(\vu'_k)\ge w] \quad\text{where $\vu'_k=\FDAlg(\Sigma^{(k-1)},v_{k,2})$}.
\end{align*}
Then $0\le A\le B$ and $c,d\ge0$. If $H_k=0$, then buyer $k$ is competing, so $\vx_k$ is obtained from $r'_k$, which is equally likely to be either $v_{k,1}$ or $v_{k,2}$, while $\vz_k$ is produced from the zero valuation. If $H_k=1$, then buyer $k$ is coordinating, so $\vz_k$ is obtained from $s'_k$, which is equally likely to be either $v_{k,1}$ or $v_{k,2}$, while $\vx_k$ is produced from the zero valuation. Therefore,
\[
    \E[y_{j,w}(\vx_{1..k})-y_{j,w}(\vx_{1..k-1})\mid H_1,\dots,H_{k-1}]=\frac{1-q}{2}\left(f(A+c)-f(A)+f(A+d)-f(A)\right),
\]
whereas
\[
    \begin{aligned}
        &\E[y_{j,w}(\vx_{1..k}\cup\vz_{1..k})-y_{j,w}(\vx_{1..k-1}\cup\vz_{1..k-1})\mid H_1,\dots,H_{k-1}]\\
        &\qquad=\frac{1}{2}\left(f(B+c)-f(B)+f(B+d)-f(B)\right).
    \end{aligned}
\]
By the concavity of $f$ and the fact that $A\le B$, each marginal increment at $A$ is at least the corresponding marginal increment at $B$. Hence the first expected increment is at least $1-q$ times the second.

Taking expectations, summing this inequality over $k$ from $1$ to $n$, and telescoping proves the stronger inequality for the fixed random order and internal randomness. Averaging over $\pi$ and the internal randomness of $\FDAlg$, then summing over $j$ and integrating over $w$, gives
\[
    \E[\FD(\vx_1,\dots,\vx_n)]\ge(1-q)\E[\FD(\vz_1,\dots,\vz_n)],
\]
as desired.
\end{proof}

\section{Deferred Proofs from \texorpdfstring{\Cref{sec:free-disp}}{Section~\ref*{sec:free-disp}}} \label{app:free-disp}

\begin{lemma}\label{lem:partition-greedy-smooth}
    Consider the partition-matroid instance in \Cref{thm:random-order-partition-greedy}. Independently retain each part with probability $q$, and run greedy on the retained parts in a uniformly random order. If $A$ is the resulting set and $O$ is an optimal solution to the original instance, then
    \[
        \E[f(A)]\ge \left(q-\frac{q^2}{2}\right)f(O).
    \]
\end{lemma}

\begin{proof}
    Give every part $P_i$ an independent arrival time $T_i$ drawn uniformly from $[0,1]$, and let $A(t)$ be the greedy solution after processing all parts with arrival time at most $t$. For every fixed $q$, the set of parts with $T_i\le q$ retains each part independently with probability $q$, and their order is uniformly random. Thus, $A(q)$ has the same distribution as the set $A$ in the statement.

    By monotonicity, we may assume that $O$ contains exactly one element $o_i$ from each part $P_i$. Let $O(t)=\{o_i:T_i>t\}$ be the optimal elements whose parts have not arrived by time $t$. We first claim that, for every realization of the arrival times and every $t$,
    \[
        f(A(t))+f(A(t)\cup O(t))\ge f(O).
    \]
    Indeed, consider the arrival of a part $P_i$, and write $A^-$ and $A^+$ for the greedy sets immediately before and after this arrival, respectively. Let $O^-$ be the set of unarrived optimal elements immediately before the arrival and let $O^+=O^-\setminus\{o_i\}$. If $u_i$ is the element chosen by greedy, then
    \begin{align*}
        f(A^+)-f(A^-)
        &=f(u_i\mid A^-)\ge f(o_i\mid A^-)\ge f(o_i\mid A^-\cup O^+)\\
        &=f(A^-\cup O^-)-f(A^-\cup O^+)\\
        &\ge f(A^-\cup O^-)-f(A^+\cup O^+),
    \end{align*}
    where the first inequality is the greedy choice, the second is submodularity, and the last is monotonicity. Hence, $f(A(t))+f(A(t)\cup O(t))$ cannot decrease at an arrival. Its initial value is $f(\emptyset)+f(O)\ge f(O)$, proving the claim.

    Let $h(t)=\E[f(A(t))]$. Condition on the process through time $t<1$. Every unarrived part has instantaneous arrival rate $1/(1-t)$. If $P_i$ arrives next, greedy gains at least $f(o_i\mid A(t))$. Consequently, for almost every $t<1$,
    \begin{align*}
        h'(t)
        &\ge \frac{1}{1-t}\E\left[\sum_{i:T_i>t} f(o_i\mid A(t))\right]\\
        &\ge \frac{1}{1-t}\E\left[f(O(t)\mid A(t))\right]\\
        &\ge \frac{f(O)-2h(t)}{1-t}.
    \end{align*}
    The second inequality follows from submodularity, and the last follows from the claim because $f(O(t)\mid A(t))=f(A(t)\cup O(t))-f(A(t))$.

    Multiplying this differential inequality by $(1-t)^{-2}$ gives
    \[
        \frac{d}{dt}\left(\frac{h(t)}{(1-t)^2}\right)
        \ge \frac{f(O)}{(1-t)^3}.
    \]
    Since $h(0)=f(\emptyset)\ge 0$, integration from $0$ to $q<1$ yields
    \[
        h(q)\ge (1-q)^2h(0)+\frac{1-(1-q)^2}{2}f(O)
        \ge \left(q-\frac{q^2}{2}\right)f(O).
    \]
    The case $q=1$ follows by continuity.
\end{proof}

\begin{proof}[Proof of \Cref{thm:lb-xos-adv-order}]
    The hardness follows from the hardness result for online submodular welfare maximization ($\OSWM$) by \cite{kapralov2013online}. In the submodular welfare maximization problem, there are $n$ buyers and $m$ items, and each buyer $i$ has a submodular valuation function $f_i : 2^{[m]} \to \R_{\ge 0}$. The algorithm seeks to allocate the $m$ items to the $n$ buyers in a way that maximizes $\sum_{i=1}^n f_i(S_i)$, where the $S_i$ are the disjoint sets of items allocated to the buyers. In the online setting, items arrive in an adversarial order, and the algorithm can see the buyers' valuations only for the items that have arrived so far. It must then irrevocably assign each arriving item to a buyer.
    
    We first state their hardness result, which holds even for a strict subclass of submodular valuations called \emph{coverage valuations}.

    \begin{definition}[Coverage valuation functions]
        A valuation function $f : 2^{[m]} \to \R_{\ge 0}$ is called a \emph{coverage valuation} if there is a set family $\{A_j\}_{j \in [m]}$ such that $f(S) = |\bigcup_{j \in S} A_j|$ for all $S \subseteq [m]$.
    \end{definition}

    \begin{theorem}[\cite{kapralov2013online}] \label{thm:lb-OSWM}
        Unless $\textnormal{\textsc{NP}} = \textnormal{\textsc{RP}}$, there is no $(1/2 + \eps)$-competitive polynomial-time algorithm (even randomized, against an oblivious adversary) for the online welfare maximization problem with coverage valuations, for any constant $\eps > 0$.
    \end{theorem}

    Let the instance of $\OSWM$ used in the theorem above be defined by the items $1, \dots, m$ that arrive in this order, the buyers $1, \dots, n$, and the set families $\mathcal{A}_i = \{A_{i, j}\}_{j \in [m]}$ for all $i \in [n]$. Note that $\mathcal{A}_i$ completely describes the coverage valuation $f_i$ for buyer $i \in [n]$. For all $i \in [n]$, let $U_i = \bigcup_{j=1}^m A_{i, j}$ be the universe containing the set family $\mathcal{A}_i$. Without loss of generality, assume that $U_1, \dots, U_n$ are disjoint, since the buyers' valuations are independent of one another.
    
    We now embed such an instance into an online free-disposal allocation ($\OFDA$) instance as follows.
    \begin{itemize}
        \item \textbf{The offline items} are the disjoint union of the universes: $U = \bigcup_{i \in [n]} U_i$.
        \item \textbf{The online buyers} correspond to the online items in $\OSWM$. Specifically, there are $m$ buyers, arriving in the order $1, \dots, m$.
        \item When buyer $i \in [m]$ arrives, their \textbf{valuation function} $v_i$ over any subset $S \subseteq U$ is
        \[v_i(S) = \max_{j \in [n]} |S \cap A_{j, i}|.\]
        Note that this is obviously an XOS valuation: the $j$-th additive clause has value $1$ for items in $A_{j, i}$, and $0$ otherwise.
    \end{itemize}

    The reduction above is a polynomial-time reduction: note that $v_i$ can be succinctly described by the sets $A_{i, j}$. We will now show that
    \begin{itemize}
        \item Any (online) allocation for $\OFDA$ can be mapped (online) in polynomial time to a feasible (online) allocation for $\OSWM$ with at least as much welfare, and
        \item Any (online) feasible allocation for $\OSWM$ can be mapped (online) in polynomial time to an (online) allocation for $\OFDA$ with the same free-disposal value.
    \end{itemize}

    These two conditions are sufficient to establish our hardness result using \Cref{thm:lb-OSWM}. First, the two conditions imply that the two instances have equal hindsight optima. Then, if there were an efficient online algorithm that was $(1/2 + \eps)$-competitive for $\OFDA$, there would be an efficient online algorithm that was $(1/2 + \eps)$-competitive for $\OSWM$, but this is impossible.
    
    To prove the claims above, the intuition is that, in $\OFDA$, buyer $i$'s choice of a subset $S$ corresponds to choosing the additive clause $j \in [n]$ that supports $v_i(S)$; this, in turn, corresponds to allocating item $i$ to buyer $j$ in $\OSWM$.

    \paragraph{$\OFDA$ to $\OSWM$.}
    For a partial allocation $S_1, \dots, S_k$ through buyer $k \in [m]$ in $\OFDA$, let $j_i \in [n]$ be the additive clause that supports $v_i(S_i)$. Correspondingly, in $\OSWM$, we partially assign items $1, \dots, k$ to buyers $j_1, \dots, j_k$, respectively, and let $T_1, \dots, T_n$ denote the sets of items allocated to each of the $n$ buyers so far. Note that this map can be implemented online and efficiently. Since, in $\OFDA$, the additive clause $j_i$ supports $v_i(S_i)$,
    \begin{align*}
        \FD(S_1, \dots, S_k) &= \sum_{t \in U} \max_{i \in [k]} \bm{1}[t \in S_i \cap A_{j_i, i}] \\
        &\le \sum_{t \in U} \max_{i \in [k]} \bm{1}[t \in A_{j_i, i}] \\
        &= \sum_{j=1}^n \sum_{t \in U_j} \max_{i \in [k]} \bm{1}[t \in A_{j_i, i}] \\
        &= \sum_{j=1}^n \left|U_j \cap \bigcup_{i \in [k] : j_i = j} A_{j, i}\right| \\
        &= \sum_{j=1}^n \left|\bigcup_{i \in T_j} A_{j, i}\right| = \sum_{j=1}^n f_j(T_j)
    \end{align*}
    where the second equality holds because $U_1, \dots, U_n$ are disjoint universes whose union is $U$, while the third and fourth equalities hold because, by definition and disjointness, $U_j$ intersects only $A_{j, i}$.

    \paragraph{$\OSWM$ to $\OFDA$.}
    In the derivation above, the only inequality occurs because $S_i \cap A_{j_i, i} \subseteq A_{j_i, i}$. Therefore, for a partial feasible assignment of items $1, \dots, k$ to buyers $j_1, \dots, j_k$ in $\OSWM$, if we allocate the set $S_i = A_{j_i, i}$ to buyer $i \in [k]$ in $\OFDA$, then equality holds throughout the derivation, proving our claim.
\end{proof}

\section{Deferred Proofs from \texorpdfstring{\Cref{sec:sssi}}{Section~\ref*{sec:sssi}}} \label{app:sssi}

\begin{proof}[Proof of \Cref{lem:lb-concave-item}]
For $a \ge b > 0$, define the concave valuation
\[
    \phi_{a,b}(x)
    :=a\min\left\{x,\frac12\right\}
      +b\left(x-\frac12\right)^+.
\]
Thus, its first and second half-units have marginal values $a$ and $b$, respectively. Fix $M > 10$ and $0 < \eps < 1/8$. There are three buyers, numbered $1, 2, 3$, whose low and high valuations are
\[
\begin{array}{c|c|c}
    \text{buyer}&L_i&H_i\\ \hline
    1&\phi_{4\eps,3\eps}&\phi_{2,1}\\
    2&\phi_{8\eps,2\eps}&\phi_{M,4} \\
    3&\phi_{6\eps,5\eps}&\phi_{M-1,\eps}
\end{array}
\]
For every buyer $i$, the two possibilities in the table form her fixed Game of Googol pair. An independent fair coin assigns one possibility to $s_i$ and the other to $r_i$. The buyers arrive in the fixed order $1, 2, 3$. Write $z_i = \mathtt{0}$ when $(s_i, r_i) = (L_i, H_i)$ and $z_i = \mathtt{1}$ when $(s_i, r_i) = (H_i, L_i)$.

We use the strict threshold comparison in the natural algorithm, consistent with \Cref{alg:sspi-frac-knapsack}. No tie-breaking is required in this instance, as the twelve marginal values occurring across the six possibilities are all distinct:
\[
    M>M-1>4>2>1>8\eps>6\eps>5\eps>4\eps>3\eps>2\eps>\eps.
\]

It is convenient to double all rewards so that a selected half-unit of marginal value $a$ contributes $a$. The sample optimum consists of the two largest sampled marginals, and the threshold equals each of these two values with probability $1/2$. Similarly, the prophet's doubled reward is the sum of the two largest real marginals. Enumerating the eight equally likely orientation vectors gives the following table; the two displayed thresholds in each row are equiprobable.
\[
\begin{array}{c|c|c|c}
    \vz&\tau_s&
    \E[2\Alg(\vs,\vr)\mid \vz]&2\Opt(\vr)\\ \hline
    \mathtt{000}&(8\eps,6\eps)&3&2M-1\\
    \mathtt{001}&(M-1,8\eps)&\frac{M+3}{2}&M+4\\
    \mathtt{010}&(M,4)&\frac{M-1}{2}&M+1\\
    \mathtt{011}&(M,M-1)&0&3\\
    \mathtt{100}&(2,1)&M+4&2M-1\\
    \mathtt{101}&(M-1,2)&M+2&M+4\\
    \mathtt{110}&(M,4)&\frac{M-1}{2}&M-1+8\eps\\
    \mathtt{111}&(M,M-1)&0&14\eps
\end{array}
\]
For example, in orientation $\mathtt{000}$, all samples are low and all real valuations are high. Buyer $1$ therefore consumes the entire item at either threshold and obtains doubled reward $2 + 1 = 3$, whereas the prophet takes the first marginals of buyers $2$ and $3$ for a doubled reward of $2M - 1$. In orientation $\mathtt{001}$, the threshold is $M - 1$ or $8\eps$: the algorithm respectively gets the first marginal $M$ of buyer $2$ or both marginals of buyer $1$, for a conditional expected doubled reward of $(M + 3)/2$. The remaining rows follow in the same way.

Averaging the last two columns gives
\[
    2\Alg=\frac{7M+19}{16}
    \qquad\text{and}\qquad
    2\Opt=\frac{8M+9+22\eps}{8}.
\]
Consequently, the competitive ratio on this instance is
\[
    \frac{\Alg}{\Opt}
    =\frac{7M+19}{2(8M+9+22\eps)}.
\]
Finally, keeping any fixed $\eps\in(0,1/8)$ and taking $M\to\infty$ makes the ratio converge to
\[
    \frac{7}{16}=0.4375 < \frac{1}{2}.
\]
Thus, the competitive ratio of the natural generalization in the single-sample setting is at most $7/16$.
\end{proof}

\begin{proof}[Proof of \Cref{lem:grid-opt-approx}]
Fix $\vr$. Let $\vx^*$ be an allocation attaining $\Opt$. We round each coordinate up to the nearest multiple of $1/M$ by setting $\widehat{x}_i:=\frac{\left\lceil Mx_i^*\right\rceil}{M}$. Let $k_i := \left\lceil Mx_i^*\right\rceil$. Monotonicity gives
\[
    \widehat{\Opt} := \sum_{i=1}^n r_i(\widehat{x}_i)
    \ge \sum_{i=1}^n r_i(x_i^*)
    =\Opt,
\]
while feasibility of $\vx^*$ implies
\[
    K:=\sum_{i=1}^n k_i
    \le M\sum_{i=1}^n x_i^*+n
    \le M+n.
\]

Therefore, $\widehat{\Opt}$ is at most the optimal value obtained by including $K \le M + n$ marginal gains. Recall that $\Opt_M$ is exactly the optimal value obtained by including exactly $M$ marginal gains. One way to obtain $\Opt_M$ is to include the largest $M$ marginal gains in $\widehat{\Opt}$, which have value at least $\frac{M}{K} \cdot \widehat{\Opt}$. Therefore,
\[\Opt_M \ge \frac{M}{K} \cdot \widehat{\Opt} \ge \frac{M}{M + n} \cdot \widehat{\Opt} \ge \frac{M}{M + n} \cdot \Opt.\qedhere\]
\end{proof}

\newpage
\bibliographystyle{alpha}
\bibliography{refs,references}
\end{document}